\documentclass[sigconf]{acmart}
\renewcommand\footnotetextcopyrightpermission[1]{} 
\usepackage{booktabs} %
\usepackage{algorithm}
\usepackage{algpseudocode}
\usepackage{fancyvrb}
\usepackage{multicol}

\newcommand{\eat}[1]{}

\DeclareFontFamily{U}{mathx}{\hyphenchar\font45}
\DeclareFontShape{U}{mathx}{m}{n}{
      <5> <6> <7> <8> <9> <10>
      <10.95> <12> <14.4> <17.28> <20.74> <24.88>
      mathx10
      }{}

\def\E{{\ensuremath{\mathbb E}}}

\long\def\comment#1{}

\newcommand{\ical}{\ensuremath{\mathcal I}}

\setcopyright{none}

\begin{document}

	\graphicspath{ {./figs/} }
	\title{Count-Min: Optimal Estimation and Tight Error Bounds using Empirical Error Distributions}
	
	\author{Daniel Ting}
      \affiliation{
Tableau Software \\
1621 N 34th Street \\
Seattle, WA
	}
	\email{dting@tableau.com}

	\begin{abstract}
	The Count-Min sketch is an important and well-studied data summarization method. It allows one to estimate the count of any item in a stream using a small, fixed size data sketch. However, the accuracy of the sketch depends on characteristics of the underlying data. This has led to a number of count estimation procedures which work well in one scenario but perform poorly in others. A practitioner is faced with two basic, unanswered questions. Which variant should be chosen when the data is unknown? Given an estimate, is its error sufficiently small to be trustworthy?
	
	We provide answers to these questions. We derive new count estimators, including a provably optimal estimator, which best or match previous estimators in all scenarios. We also provide practical, tight error bounds at query time for both new and existing estimators. These error estimates also yield procedures to choose the sketch tuning parameters optimally, as they can extrapolate the error to different choices of sketch width and depth.
	
	The key observation is that the distribution of errors in each counter can be empirically estimated from the sketch itself. By first estimating this distribution, count estimation becomes a statistical estimation and inference problem with a {\em known} error distribution. This provides both a principled way to derive new and optimal estimators as well as a way to study the error and properties of existing estimators. 
	\end{abstract}
	
	\maketitle
	
	\section{Introduction}
	The Count-Min sketch has proven to be one of the most effective sketches for obtaining approximate counts for pointwise queries and for computing approximate inner products. It is especially effective in the common data scenario where the count distribution is highly skewed. 
	
		However, there are notable cases where the sketch performs sub-optimally or poorly. For example, when there are few heavy hitters and a large number of items relative to the size of the sketch, the Count-Min sketch is highly biased and performs poorly compared to the Count sketch \cite{charikar2002countsketch}. This has led to a number of attempts
		\cite{jin2003dynamically}, \cite{lee2005improving}, \cite{deng2007new}, \cite{lu2008counter}, \cite{chen2017bias} to improve estimation from the Count-Min sketch in these regimes.
		In all cases but one \cite{lu2008counter}, these methods can be shown to perform worse than the basic Count-Min estimator in some regimes or for some sketch parameter settings. The one case with guaranteed better accuracy, however, can only be applied in the highly restrictive and computationally expensive setting where all possible items are known and their counts jointly estimated. As a result, it is unclear to a practitioner which method to choose. Although several empirical studies \cite{rusu2007statistical}, \cite{cormode2008finding} have attempted to address this issue, choosing the best method has required a priori knowledge of the properties of the unseen data.
		
		A second issue with the Count-Min sketch is that although it has a probabilistic error guarantee, this guarantee is extremely loose and of no practical use when reporting the error of any query.  Again, the only proposed method for obtaining  errors with practical magnitudes is given by \cite{lu2008counter} where all counts must be decoded. 		

	This paper introduces methods that provides better accuracy under all regimes and takes the guesswork out of count estimation. The resulting estimator also has a tight, practical error bound. Furthermore, it can utilize joint estimation of multiple counts to yield more accurate results without needing to know the entire universe of items. 

	Our approach treats count estimation from the Count-Min sketch as a statistical estimation problem where the irrelevant counts are modeled as error terms. 
	The key idea is that the distribution of these error terms can be estimated from the sketch itself. Equipped with an error distribution, we consider two classes of estimators: ones which use the full likelihood information and  ad-hoc estimators with some good properties. All existing estimators are shown to be from the latter class. 
	For these estimators, we show that bootstrap methods can be used to debias a wide class of estimators and obtain tight confidence intervals that bound the error. 

	We propose two likelihood based estimators: the standard maximum likelihood estimator and a Bayesian estimator. The Bayesian estimator, while more computationally expensive, is proved to be optimal even when the sketch is of fixed depth. The more practical maximum likelihood estimator is empirically shown to outperform all other methods in all scenarios.  
	
	Key to the likelihood based methods is a non-parametric estimate of the error distribution. We show this can be accomplished with log-concave density estimation. This estimator has attractive properties as it requires no tuning parameters and yields a concave log-likelihood function that ensures maximum likelihood estimation is fast and easy. We further show that it generates robust count estimators even when the assumption of log-concavity is false.
	
	In addition to the practical improvements motivated by theory, our work also advances our understanding of the Count-Min and related sketches. We serve as a brief survey of existing estimation algorithms and summarize the techniques used. We show that unlike existing methods which exploit only one or two techniques, our method is able to exploit all of them to obtain better results. Furthermore, we use asymptotic theory to explain under which regimes  different count estimators and sketches perform well. 

	This understanding also has practical consequences in sketch construction. In particular, we find given a fixed space constraint, it is generally preferable to reduce the number of hash functions and increase the width of the sketch, as it increases the likelihood of belonging in the "super-efficient" regime where 
	 the Min estimator achieves the optimal rate. When additional information about the error distribution is known a priori, we show how to optimally select these sketch parameters.

	Our methods may also be applied to other sketches such as the Count sketch \cite{charikar2002countsketch}, also known as the Fast-AMS sketch when applied to inner-product estimation \cite{alon1999space}.

	The paper is structured as follows. First, we review the Count-Min sketch and define the empirical error distribution relative to a pointwise query. Next, we give a brief survey of existing work on improving estimation for the sketch, provide insights into how they work, and show they can be generalized in natural ways. Section \ref{sec:bootstrap} then introduces the bootstrap and shows how simple statistics can be converted into unbiased estimators for the count and gives procedures to construct tight error bounds. As simple statistics may not make full use of the information in the data, section \ref{sec:likelihood} shows that the true likelihood can estimated from the data and proposes estimators based on it.
	We also show that the resulting estimators have robust estimation properties and that they can be used to estimate multiple counts jointly through regression. Section \ref{sec:empirical} provides empirical results on real and synthetic data to show that our estimators are indeed the most accurate in a variety of settings and that the error bounds are tight
	We then discuss asymptotics that aid our understanding of the sketch, applications to parameter tuning, and the use of our techniques to other sketches, in streaming settings, and for inner product estimation. 
	
	Throughout the paper we rely heavily on statistical estimation theory and concepts that we unfortunately do not have sufficient space to cover in detail. These concepts are the full distribution based counterparts to the tail probability and concentration inequality driven theory common in the sketching literature. 
	\vspace{-0.1cm}
	\section{Count-Min}
	The Count-Min sketch compresses and aggregates a large and possibly unknown number of $(item, count)$ tuples into a finite sketch of $r \times k$ numeric counters. 
	It allows for two basic types of queries: 1) pointwise  queries
	which provide an estimate of the aggregated count for any item or set of items,
	and 2) inner product queries which provide for an estimate for
	$u^T v$ for count vectors $u$ and $v$ indexed by distinct items. 
	We write the vector of counts indexed by item by $\mathbf{n}$.
	These two basic queries can be used to formulate more complex queries. For example, aggregated counts
	for range queries can be constructed out of pointwise queries that expand numeric valued items into membership in a set of dyadic ranges 
	\cite{cormode2005countmin}.
	We focus on pointwise queries in this paper and briefly discuss the application of our techniques to the inner product case.
	
	The Count-Min summarization technique can be decomposed into two parts: the construction of the sketch and the estimation procedure for  count queries. In this paper, we focus on improvements to estimation and not on sketch construction. For clarity, we will refer to the construction as the Count+ summarization and the estimator as the Min estimator.  Here, the plus sign represents the one-sided errors for the sketch. 
	
	An $r \times k$ Count+ summarization consists of two parts: a hash based projection and replication. The first hashes each item to one of $k$ counters. The $k$ vector of observed counters is obtained by summing the counts in each bin. The second part simply replicates this process $r$ times with independent hashes.
	$r$ and $k$ are often referred to as the depth and width of the sketch.

	More precisely, given a hash function $h$, the item, count pair $(x_i, c_i)$ updates the counter vector $\mathbf{V}$ by the update rule
	\begin{align}
	\label{eqn: CM update}
	V_{h(x_i)}^{(new)} &= V_{h(x_i)}^{(old)} + c_i.
	\end{align}
	This process is repeated $r$ times to obtain independent identically distributed (i.i.d.) vectors $\mathbf{V^{(a)}}$ using independent hashes $h^{(a)}$ for $a = 1, \ldots, r$.

	Estimation from this sketch is simple and relies on the fact that counts are non-negative. For any of the $k$-vectors $V^{(a)}$,
	the counter $V^{(a)}_{h^{(a)}(x)}$ is an upper bound on the total count $n_x$ for item $x$. The original Min-estimator for the Count-Min sketch takes the minimum over the $r$ replicates 
	\begin{align}
	\hat{N}_x &= \min_a V^{(a)}_{h^{(a)}(x)} \geq n_x.
	\end{align}

	Several simple observations can be made from this construction and estimator. 
	Only the counters that an item is hashed to contain any information about its count. Removing an item and its count from the Count+ summarization yields vectors of exchangeable error terms where the error terms are all non-negative.
	The Min estimator is biased as it cannot underestimate the count.
	More formally, for any replicate $V$,
	\begin{align}
	\label{eqn:basic error model}
	V_i^{(a)} &= n_x 1(h^{(a)}(x) = i) + \epsilon^{(a)}_i 	
	\end{align}
	where the $\epsilon^{(a)}_i \geq 0$ are identically distributed and exchangeable.

	These observations motivate our basic strategy. Take counters which only contain error terms. Use them to empirically estimate a  non-centered, non-negative error distribution. An item's counters plus the error distribution for those counters provides all the available information to estimate the item's count.
	Apply statistical estimation techniques to estimate the count and obtain an error estimate. When the error distribution is correct, the resulting estimator is optimal.

	\begin{table}
		\begin{tabular}{c|l} 
			Symbol & Definition \\ \hline
			$\mathbf{n}$ & Vector of all counts indexed by item\\
			$\hat{n}_x$ & Estimated count for item $x$ \\
			$d$ & Number of distinct items\\
			$\ical$ & Set of indices that $x$ or $S$ are hashed to \\
			$r$ & Number of replicates in Count-Min sketch\\
			$k$ & Number of counters in one replicate \\
			$h^{(a)}$ & Hash function for replicate $a$ \\
			$\mathbf{V}$ & Count-Min counters \\
			$V^{(a)}_i$, $V_{(a,i)}$ & $i^{th}$ counter in replicate $a$ \\
			$\epsilon$ & Vector of errors (relative to some item $x$) \\
			$F$, $\mathbb{F}$ & True and empirical distribution of errors  \\
			$\mathbf{M}, M^{(a)}$ & Projection matrix for the sketch and for  replicate $a$ \\
			$\lambda$ & Expected number of items per counter $\lambda = d/k$.
		\end{tabular}
	\caption{Table of symbols}
	\end{table}

	\subsection{Linear algebra of the Count-Min sketch}
	\label{sec:linear alg}
The Count+ summarization is an example of a linear sketch. In other words, each replicate is a random projection $M^{(a)}$ of the counts $\mathbf{n}$ where the construction of $M^{(a)}$ does not depend on $\mathbf{n}$. This may be expressed as
\vspace{-0.1cm}
\begin{align}
\label{eqn:random projection}
V^{(a)} &= M^{(a)}\mathbf{n}
\end{align}
where $M^{(a)}$ is a $k \times d$ random binary matrix with precisely 1 non-zero value per column.
More explicitly, $M^{(a)}_{ix} = 1$ if $h^{(a)}(x) = i$ and $0$ otherwise.
For succinctness in notation we denote the concatenation of the $V^{(a)}$ as simply $\mathbf{V}$ and likewise for $\mathbf{M}$.
We also write $V^{(a)}_i$  by $\mathbf{V}_{(a, i)}$ and similarly for $\mathbf{M}$.

Whenever only a subset $S$ of items are of interest, 
the sketch has the form,
\vspace{-0.2cm}
\begin{align}
\label{eqn:CountMin regression}
\mathbf{V} &= \mathbf{M}_{\cdot, S}  \mathbf{n}_S + \epsilon(S) \\
\epsilon(S) &= \mathbf{M}_{\cdot, S^c} \mathbf{n}_{S^c}.
\end{align}
The equation representing the counters $\mathbf{V}$ has the same form as a linear regression problem where $\mathbf{M}_{\cdot, S}$ are the known covariates and $\mathbf{n}_S$ are the unknown regression coefficients.
The error terms $\epsilon(S)$ are defined relative to the queried items $S$.
It differs slightly from typical linear regression problems in that the errors are not centered to have mean zero, and the distribution of the errors is not known or assumed.
For notational convenience, we will simply write $\epsilon$ for the error term as $S$ is always clear from the context. 

\subsection{Empirical distributions}
Given an item $x$ and Count+ summary, only the $r$ counters that $x$ hashes to provide information about the count $n_x$. The remaining $r(k-1) \gg r$ counters are draws from an error distribution. This large sample allows the error distribution to be accurately estimated and reduces the count estimation problem to a familiar problem of parameter estimation with a known error distribution. 

Denote the unknown true error distribution's cumulative distribution function (c.d.f.) as $F$ and its density or mass function as $f$. When $Y$ is drawn from a distribution with c.d.f. $F$, we write $Y \sim F$.
In the case of a pointwise query for a single item, the distribution of a counter $V_{(a, h^{(a)}(x))} \sim F(\cdot - n_x)$. Estimating the count $n_x$  is a parametric estimation problem from the one-parameter location family $\{F(\cdot - \theta)\}_{\theta \geq 0}$ of distributions.

	\section{Existing work}
		Several existing improvements to the Min estimator have been proposed. 
	The estimation techniques for the Count+ summary can be categorized into four basic ideas:
	\begin{enumerate}
		\item Bias reduction
		\item Linear Regression
		\item Support constraints
		\item Robust objective choice
	\end{enumerate}
	Each existing estimator exploits only one or two of these ideas. For example, the Min estimator exploits only the non-negative support of the error distribution. The Median estimator exploits only a robust $L_1$ objective choice.
	
	\subsection{Debiasing}
	\label{sec:debiasing}
	Most prior work, \cite{deng2007new}, \cite{jin2003dynamically}, \cite{chen2017bias}, focuses on debiasing the estimator under different choices of objectives. 
	We describe this debiasing operation with a more general procedure and list the choices made by each procedure. This allows us to extend debiasing to a large class of base estimators, such as any quantile.
	
	Let $\ical$ be the set of (replicate number, index) $\in \{1,\ldots,r\} \times \{1,\ldots, k\}$ pairs that item $x$ is hashed to. Let $T$ be some  function on a set of $r$ counters so that
	\begin{align}
	\label{eqn:translation property}
	T(V_\ical) &= n_x + T(\epsilon_\ical).
	\end{align}
	We refer to this as the translation property in this paper.
	Obvious examples of $T$ include the mean, minimum, median, and any quantile. These are also all special cases of maximizers of the form $T(V_\ical) = \arg \max_{\theta} J(V_\ical - \theta)$. For the mean, $J(x) = \|x\|_2^2$, and for the median, $J(x) = \|x\|_1$ and is a robust loss function.
	
	For any $T$ satisfying this property, $T(V_\ical) - \mu$ is an unbiased estimate for $n_x$ when $\mu = \E T(\epsilon_\ical)$. 
	This yields a general method for constructing a debiased estimator. 1) Choose a function $T$ with the translation property, and 2) find an empirical estimate of the bias $\mu$.
	
	For the hCount* estimator \cite{jin2003dynamically}, $T$ remains the minimum. To estimate the bias, they explicitly query for a small set of items that are known to have count $0$ and take the average of the corresponding estimates. 

	For the CMM estimators \cite{deng2007new}, $T$ is taken to be the median. Rather than explicitly querying to find noise counters, they use counters that do not to contain the query key to estimate the bias. 
	 Since $\E T(\epsilon_\ical) \approx \E T(\epsilon_{\ical'})$ regardless of the sizes of $\ical$ and $\ical'$, the resulting estimate is nearly unbiased. 

	 Bias Aware estimation \cite{chen2017bias} proposes other debiased Median and Mean estimators for $T$. They differ from other debiasing methods since they use information not contained in the sketch itself.
	 Rather than directly applying the mean or median to the set $V_{\ical}$ of relevant counters, they compute "debiased counters" $\tilde{V}_i = V_i - \beta (W_i - 1)$ where $W_i$ is the number of items hashed to counter $V_i$ and $\beta$ is a per item bias estimate. The statistic $T(\tilde{V}_\ical)$ has the translation property and does not need further debiasing. However, computing this requires knowing and being able to iterate over the universe of distinct items.

	\subsection{Regression and Support Constraints}	
	When multiple items counts are estimated together, estimation can be improved.
	One item's estimate can reduce the error for another item when there is a hash collision. More formally, equation \ref{eqn:CountMin regression} shows that adding elements to the set $S$ of desired item counts reduces the number of items mapping to the error term. When the added items are heavy hitters, this can substantially reduce the magnitude of the error.
	The choice of regression model is thus dictated by what one knows about the universe of items and assumptions about the unknown error distribution
	
	Under the assumption that the error distribution is normal and only a subset $S$ of items are known, one recovers the linear least squares method of \cite{lee2005improving}. This is equivalent to the solution of the maximization problem
	\begin{align}
	\hat{n}_s = \arg \max_\theta \|V - \mathbf{M}_{\cdot, S} \theta\|_2^2.
	\end{align}
	In the case where all item counts are jointly estimated and the linear system $\mathbf{V} = \mathbf{M} \theta$ is overdetermined,
	the least-squares estimator finds the exact counts.

	If the entire universe of items is known, the Counter Braids estimation algorithm \cite{lu2008counter} is guaranteed to be no worse than the Min estimator and can often recover the exact counts. 
	The Counter Braids estimator does so via a message passing algorithm that provides deterministic upper and lower bounds on the estimated counts. We show in appendix \ref{sec:appendix counter braids}
	that this algorithm can be formulated in as a standard optimization problem. It is a cutting plane algorithm \cite{kelley1960cutting} for finding the feasible set for an optimization problem, and that the feasible set exploits the non-negative support of error distributions. 

	Exploiting ideas from both methods yields the general class of regression based procedures that solve the {\em constrained} optimization problem
\vspace{-0.2cm}
	\begin{align}
	\min_{\theta > 0\,s.t.\, \mathbf{M}_{\ical, \cdot} \theta \leq \mathbf{V}_\ical} J(\mathbf{V}_\ical - \mathbf{M}_{\ical, \cdot} \theta)
	\end{align}
	where $J$ is some loss function. 
	Section \ref{sec:likelihood} will show that an estimated log-likelihood function yields a good loss function. 

	\section{Our methods}
	When the problem is fulled modeled by a statistical model, the four techniques listed in the previous section can be simplified into two: linear regression and modeling the error distribution. The error distribution encodes the bias, support, and optimal objective function to use for count estimation. 
	In addition, knowledge of the error distribution yields the exact sampling distribution of an estimator and corresponding tight confidence intervals (CIs).
		
	We propose two methods based on non-parametric modeling of the error distribution.  First, we propose a class of bootstrap estimators. This class of estimators can be based off statistics that are fast and easy to compute and implement. It covers all existing debiased estimators and allows for the easy generation of others such as estimators based on other quantiles or trimmed means. Second, we propose full likelihood based estimators based on an empirical estimate of the error density or mass function. These methods can incorporate regression techniques to exploit information about the universe of items. 

	\section{Bootstrap Estimators}
	\label{sec:bootstrap}
	Debiasing and computing tight error bounds bounds requires knowing the distribution of the statistic $T(V_\ical)$. 
	The bootstrap \cite{efron1979bootstrap} estimates this distribution by resampling observations and examining the distribution of the results on the simulated samples. The n\"aive bootstrap will not work since there are only a small number $r$ of relevant counters to resample. However, when a statistic has the translation  property, one can instead sample from the $r (k-1)$ error counters.

	Theorem \ref{thm:bootstrap estimate} shows that when this is done, then any $T$ with the translation property can be turned into an unbiased estimator of an item's count. 
	Existing debiased estimators can be seen as instances of this bootstrapping procedure. 
	While our analysis suggests easier ways to compute the bias and yields new estimators, our primary contribution is applying the bootstrap to yield tight confidence intervals and in its application to new base statistics $T$. We also address computational issues that arise with the bootstrap and show  biases and confidence intervals for estimators based on the minimum value or any quantile can be recovered without resorting to an expensive Monte Carlo simulation.

	\begin{theorem}
		\label{thm:bootstrap estimate}
		Let $T$ be any function that satisfies the translation property. Consider an item $x$ and the collection of indices $\ical(x)$ that $x$ is hashed to. Consider the empirical distribution of the counters excluding those in $\ical(x)$, and denote expectation under this distribution by $\mathbb{E}_{\ical(x)^c}$. Let $Y_r$ be $r$ i.i.d. draws from this distribution. Then,
		\begin{align}
		\hat{n}_x &= T(V_{\ical(x)}) - \mathbb{E}_{\ical(x)^c} T(Y_r)
		\end{align}
		 is an unbiased estimator for the count $n_x$.
	\end{theorem}
	\begin{proof}
		Let $y$ be a randomly chosen item with count $n_y = 0$.
		Denote by $\epsilon = \mathbf{V} - \mathbf{M}_{\cdot, x} n_x$ the vector of error terms for item $x$. By symmetry, $\E T(\epsilon_{\ical(x)}) = \E T(\epsilon_{\ical(y)})$.
		Since $\epsilon_{\ical(y)} = V_{\ical(y)}$ whenever $\ical(y) \cap \ical(x) = \emptyset$, $\mathbb{E}_{\ical(x)^c} T(Y_r) = \E T(\epsilon_{\ical(y)})$. Hence, 
		$\E \hat{n}_x = n_x + \E T(\epsilon_{\ical(x)} - \E T(\epsilon_{\ical(y)}) = n_x$.  
	\end{proof}

	While this theorem constructs an unbiased estimator out of any base statistic $T$ that satisfies the translation property, it is possible for the resulting estimate to be negative. 
	When counts are always non-negative, it is sensible to truncate the estimate at 0 to ensure all estimates are non-negative as well. 
	This results in a slightly biased estimator. We apply this truncation to all estimators, and hence refer to them as debiased and not unbiased estimators. We also note that the base statistic $T$ cannot be a truncated statistic. Otherwise, it cannot have the translation property.
	
\subsection{Tight error estimation}
	 Theorem \ref{thm:CI flipped} shows the bootstrap can be used to construct confidence intervals that have the correct finite sample coverage in all situations. A trivial corollary shows the resulting confidence intervals are tight. A confidence interval $R(V)$ for $n_x$ at level $1-\alpha$ is a probabilistic error bound which guarantees that 
	 $P(n_x \in R(V)) \geq 1-\alpha$. If expressible as $R(V) = (\tilde{n}_x - \delta, \tilde{n}_x + \delta)$, the error guarantee is of the form $P(|n_x - \tilde{n}_x| > \delta) < \alpha$.

	\begin{theorem}
		\label{thm:CI flipped}
		Let $u_q$ be the $q$ quantile of the empirical distribution $\mathbb{G}$ of $T(Y_r)$. The interval $[T(V_{\ical(x)}) - u_b, T(V_{\ical(x)}) - u_{a}]$ is a $(b-a)$ confidence interval for the count $n_x$. The coverage of the interval is $\mathbb{G}(u_b) - \mathbb{\underline{G}}(u_a)$ where $\mathbb{\underline{G}}(y)$ denotes the probability a draw from the empirical distribution is strictly less than $y$ rather than less than or equal to $y$.
	\end{theorem}
	\begin{proof}
		By symmetry of the errors, $P(T(\epsilon_{\ical(x)}) \in [u_{a}, u_{b}]) = \E( \mathbb{G}(u_b) - \mathbb{\underline{G}}(u_a)) \geq (b-a)$. 
		Substituting $T(\epsilon_{\ical(x)}) = T(V_{\ical(x)})  - n_x$ and rearranging gives the desired result. 
	\end{proof}

\begin{corollary}
	\label{thm:CI tight}
Any shorter interval $[T(V_{\ical(x)}) - \tilde{u}_b, T(V_{\ical(x)} - \tilde{u}_{a}]$
with $[\tilde{u}_a, \tilde{u}_b] \subsetneq [u_a, u_b]$ has coverage strictly less than $b-a$.
\end{corollary}
	\begin{proof}
		$\mathbb{G}(\tilde{u}_b) - \underline{\mathbb{G}}(\tilde{u}_a) < (b-a)$		
	\end{proof}

   Past theoretical work effectively derives an extremely loose power calculation. It finds a sample size $r \times k$ that guarantees an error less than a desired "effect size" $\delta$ with probability at least $1-\epsilon$ for some constants $\delta, \epsilon$. The precise constants needed to compute the needed sketch size depend on the unknown true counts $\mathbf{n}$. 
A simple search procedure allows one to convert power calculations 
which map $(\delta, \epsilon) \mapsto (r,k)$, 
to confidence intervals which map $(r,k,\epsilon) \to \delta$ and vice versa.
In section \ref{sec:tuning}, we examine the power calculations for our tight confidence intervals to find optimal settings of the sketch parameters.

	This confidence bound corresponding to the existing theory is obtained using Markov's inequality.
	$P(\hat{n}_x^{Min} > n_x + c \E \epsilon_1) = 
	P(\epsilon_1 > c \E \epsilon_1)^r 
	\leq c^{-r}.$
	Setting this equal to $\alpha$ 
	gives a $1-\alpha$ confidence interval with width $c \E \epsilon_1= n_{tot}  \alpha^{-1/r} / k $
	where $n_{tot}$ is the total count summed over all items.
	Given a fixed memory budget $B$, this can be expressed as the interval
	$(\hat{n}_x^{Min} - \alpha^{-k/B} n_{tot} / k, \hat{n}_x^{Min} ]$. If one chooses $r,k$ to optimize this interval, the interval width is $-(\log \alpha) n_{tot}/B$.

	This interval does not account for the shape of the count distribution and depends only on the total count $n_{tot}$. As a result, a single heavy hitter can result in an arbitrarily wide confidence interval, even though a vast majority of count items are not affected.
	While more refined analyses that account for the top heavy hitters have been proposed \cite{cormode2005summarizing}, \cite{cormode2012synopses}, these require knowing the heavy hitters or a strong assumption of Zipf distributed counts. Neither is easily estimated or verified from the sketch alone. In contrast, our bootstrap confidence intervals automatically account for the entire shape of the count distribution, including the heavy hitters, and does so with only knowledge that is readily available from the sketch. 

The improvement offered by tight error bounds are significant as the practical performance of the Count-Min and Count sketches often greatly differ from the theoretical bounds \cite{rusu2007statistical}, \cite{minton2014improved}. In the context of inner-product estimation, these bounds yielded error estimates that were over $10^6$ times larger than the true errors, though we see more modest differences in our pointwise queries.
Figure \ref{fig:CI widths} shows the disparity between our empirically driven confidence intervals and the existing Markov inequality based confidence intervals. Especially for heavy tailed distributions, the Markov inequality based intervals are often an order of magnitude larger than our intervals. 

It also shows the actual coverage of the estimators matches or exceeds the desired coverage. The coverage exceeds the desired coverage primarily when the intervals are narrow. In these cases, we verified that excess coverage is due to the discrete jumps in probability in a discrete distribution. Attempts to shorten the intervals yielded insufficient coverage. For example, reducing the intervals by $0.5$ on each side, and effectively turning the interval from a closed to open interval for discrete counts, reduced the empirical coverage for a $90\%$ CI for the MLE estimator from $0.91$ to a less than advertised coverage of $0.88$.
	Thus, the empirical results verify the theory which states they are tight as possible.
\begin{figure}
	\hspace{-0.2cm}	\includegraphics[width=3in, height=2in]{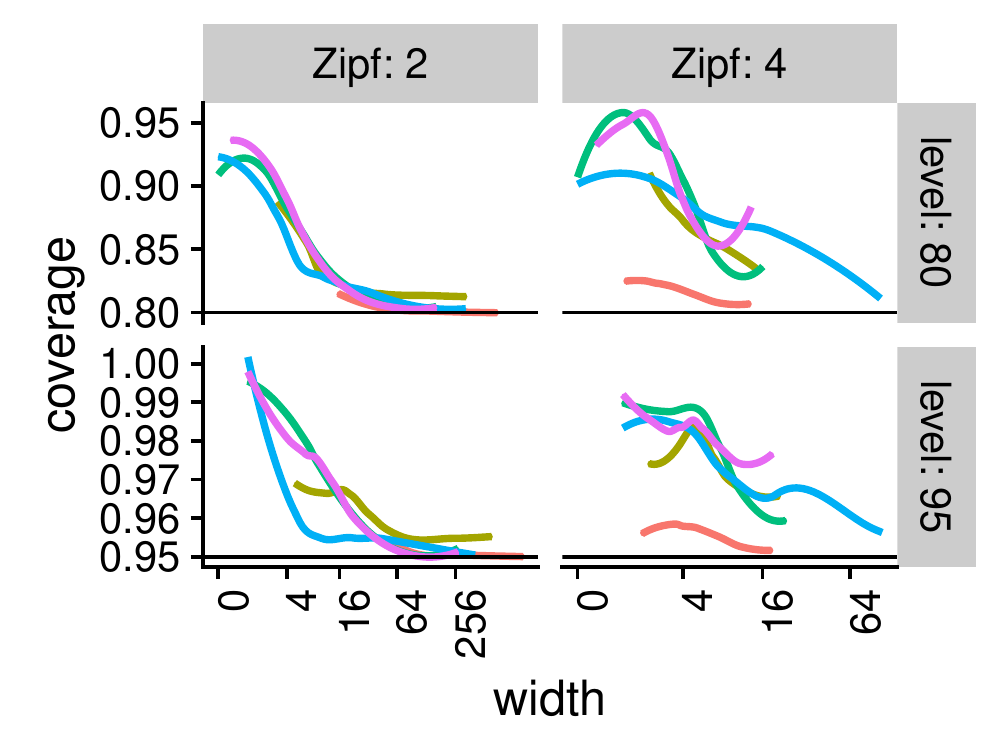} 
		\includegraphics[width=3.4in, height=2in]{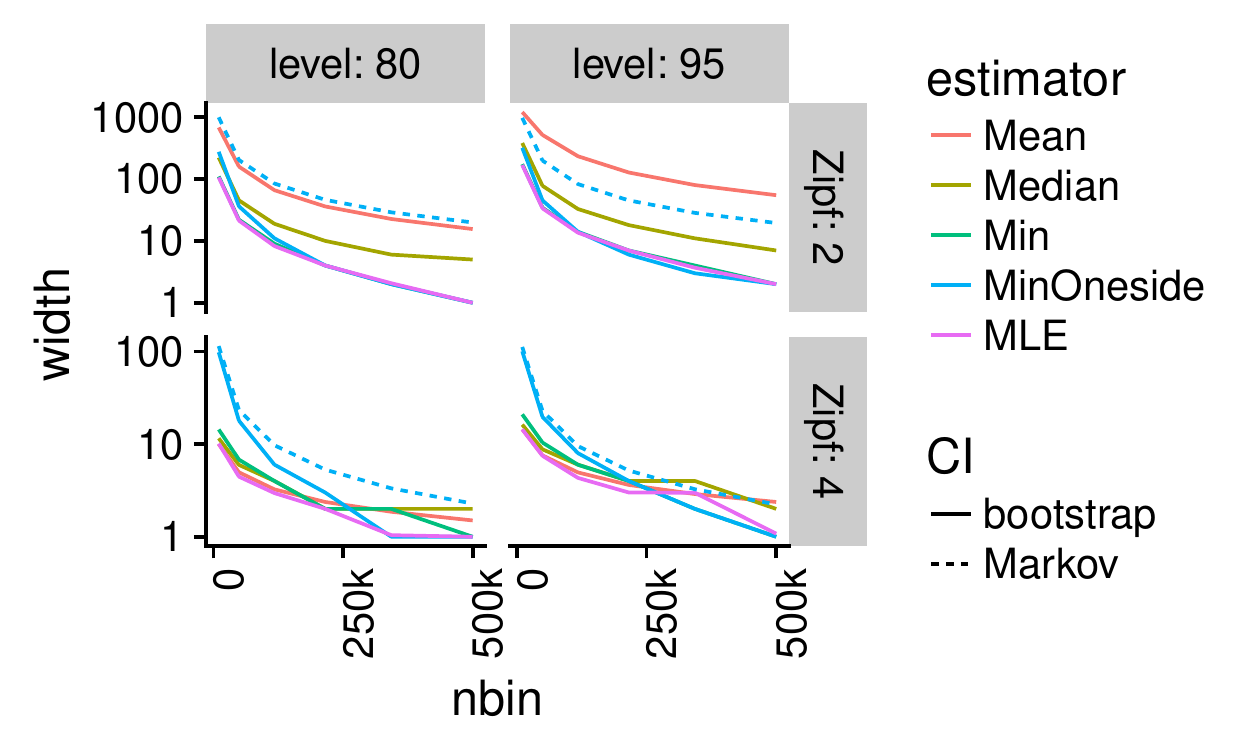}
	\vspace{-0.5cm}
	\caption{Top: The confidence intervals deliver the promised coverage. Overcoverage is due to the discrete nature of the data and probabilites when the accuracy is high. 
	Bottom: The existing probabilistic bound based on Markov's inequality is extremely poor, typically being off by an order of magnitude on a heavy tailed distribution. Except for the "one-sided" interval, all intervals are the two-sided intervals given in algorithm \ref{alg:bootstrap}.} 
		\label{fig:CI widths}
\end{figure}

	\subsection{Computation}
	Bootstrap quantities can pose some computational difficulty as they are typically calculated via Monte Carlo simulation. However, in some cases, the quantities can often be computed directly from the empirical distribution \cite{efron1994introduction}. In particular, the mean and distribution of an order statistic can be easily approximated. The order statistic $X_{(i)}$ of a set of items $X_1, \ldots, X_r$ is the $i^{th}$ smallest value in that set. For example the Min estimator is an order statistic as it is the smallest value in a set of $r$ values. 
	
	This can be done by relating the distribution of the order statistics from $F$ distributed random variables to those of $Uniform(0,1)$ random variables. 
	Recall that the inverse c.d.f. transform generates a $F$ distributed random variable from a $Uniform(0,1)$ random variable via $Y_i = F^{-1}(U_i)$ for $U_i \sim Uniform(0,1)$. Since $F$ is monotone, the order statistic $Y_{(i)} = F^{-1}(U_{(i)})$. The distribution of $U_{(i)}$ is well-known and is $U_{(i)} \sim Beta(i, r-i +1)$. 
	
	When applied to debiasing operations, this gives $\E Y_{(i)} \approx {F}^{-1}(\E U_{(i)}) = {F}^{-1}(i/(r+1))$. In particular, the Min estimator can be debiased using the estimated bias $\mu = \mathbb{F}^{-1}(1/(r+1))$ where $\mathbb{F}$ is the empirical distribution of the errors. More importantly, an {\em exact} confidence interval can be computed directly from $\mathbb{F}$ by using an outer confidence interval \cite{meyer1987outer}. 
	
	For the Min estimator, a "one-sided" $95\%$ confidence interval for the error is $[0, \mathbb{F}^{-1}(b_{95})]$ where $b_{95} = Beta_r^{-1}(0.95)$ is the $0.95$ quantile of a 
	$Beta(1,r)$ distribution. This leads to algorithm \ref{alg:debiased min} which debiases the Min-estimator and provides a confidence interval. We refer to this as a "one-sided" confidence interval since the upper bound cannot be violated.
	A two-sided interval for the Min or any quantile estimator can be similarly estimated. For the  $i^{th}$ order statistic, compute a $1-\alpha$ confidence interval $[a,b]$ for $U_{(i)}$. Theorem \ref{thm:CI flipped} gives that
	$[T(V_{\ical}) - \mathbb{F}^{-1}(b), T(V_{\ical}) - \mathbb{F}^{-1}(a)]$ is a $1-\alpha$ confidence interval for the estimate. For implementation purposes, note that $T$ is the base estimator prior to debiasing.

	Even when the bootstrap quantities cannot be directly computed from the distribution of error counters, they can be computed just once and applied to all count estimates. Since quantiles are always robust and most estimators $T$ that we consider are also robust to large errors, 
	there is little difference in estimating the bias $\mu$ and interval $[u_a, u_b]$ using all counters rather than only the counters that do not contain a given item. This yields algorithm \ref{alg:bootstrap} which debiases an estimator and returns a confidence interval. 
	
		\begin{algorithm}
		\begin{algorithmic}
			\Function{Pre-process Errors}{$T, V, a, b$}
			\For{i=1,\ldots,k}
			\State $Y_i = \{V_{(j,i)}: j=1,\ldots,r\}$ 
			\State $Z_i \gets T(Y_i)$
			\EndFor
			\State Let $\mathbb{G}$ be the empirical c.d.f. of the $\{Z_i\}$
			\State $(u_{a}, u_{b}) \gets (\mathbb{G}^{-1}(a), \mathbb{G}^{-1}(b))$
			\State \Return $\mu = \E_{\mathbb{G}} Z$ and $[u_{a}, u_{b}]$
			\EndFunction			
			\Function{Debiased-Estimator}{$x, T, V, \mu, u_{a}, u_{b}$}
			\State $\ical \gets \{(i,h^{(i)}(x)): i = 1,\ldots, r \}$
			\State $\hat{n}_{raw} \gets T(V_{\ical})$
			\State \Return $\hat{n}_x = \max\{0, \hat{n}_{raw} - \mu\}$ and $[\hat{n}_{raw} - u_{b}, \hat{n}_{raw} - u_{a}]$			
			\EndFunction
		\end{algorithmic}
	\caption{Bootstrap debiasing with Confidence Interval}
			\label{alg:bootstrap}
	\end{algorithm}

	\begin{algorithm}
		\begin{algorithmic}
			\Function{Debiased Count-Min}{$x, V, \ell$}
			\State $\ical \gets \{(a,h^{(a)}): a = 1,\ldots, r \}$
			\State $\hat{n}_{min} \gets \min_{i \in \ical} V_i$
			\State $\mu \gets k^{th} \mbox{ smallest value of } V$ (i.e. $\mathbb{F}^{-1}(1/r)$).
			\State $b \gets BetaCDF^{-1}(\ell, 1, r)$
			\State $u_b \gets (b \cdot r \cdot k)^{th} \mbox{ smallest value of } V $ (i.e. $\mathbb{F}^{-1}(b)$).
			\State \Return  $\hat{n}_x = \max\{\hat{n}_{min} - \mu, 0\}$ and CI $[\max\{\hat{n}_{min} - u_b, 0\}, \hat{n}_{min}]$.
			\EndFunction
		\end{algorithmic}
	\caption{Debiased Min estimator with Confidence Interval}
			\label{alg:debiased min}
	\end{algorithm}

   \section{Likelihood based estimation}
   \label{sec:likelihood}
   For the bootstrapped estimators, the procedures directly resample from the error distribution without estimating the distribution itself. 
   With likelihood based methods, it is necessary to estimate this distribution. By doing so, one is able to apply the statistical machinery for efficient estimation and inference. 
   
   We derive the error distribution and show how to estimate it non-parametrically and without any additional tuning parameters. This allows the easy application of maximum likelihood estimation as well as Bayes optimal estimation. Furthermore, the likelihood based approaches provide a framework for performing joint estimation of counts via regression to obtain even more accurate estimates.

	\subsection{Log-concave density estimation}		
	To ensure good performance under all possible count distributions, we use a non-parametric estimate of the error distribution. We do this under the assumption that the distribution of the log-errors are log-concave. The concavity has the added benefit that the continuous relaxation of the maximum likelihood objective is easily maximized by standard concave maximization algorithms. 
    Furthermore, unlike other non-parametric methods such as kernel density estimation, a log-concave density has a consistent maximum likelihood estimator \cite{dumbgen2009maximum} that requires no tuning of parameters such as the bandwidth.  
	
	Log-concave densities cover many common distributions. These include the Poisson, Binomial, Exponential, Normal, Negative-Binomial, among others. We remark that heavy tailed distributions with probability $f(y) \propto y^{-\alpha}$ for large $y$ have a log density or log mass function that is log-convex in the tails rather than concave.
	In this case, we compute a log-concave projection of the trimmed density which results in linearly decaying tails. As shown in section \ref{sec:robust}, the resulting objective function
	 is a {\em robust} objective which can perform well even when the assumptions are not met. It is similar to Huber's estimator which combines the quadratic loss associated with the mean estimator with the linear loss of the median or other quantiles.
	
	We further note that in many commonly used distributions where the log-concavity assumption is invalid, the density or mass function is  monotone decreasing. Though non-parametric density estimators for decreasing densities exist, it is unnecessary for the purposes of this paper. For a decreasing density with unbounded support,  the Min estimator is the MLE. We make this precise in Theorem \ref{thm:decreasing density} and in Theorem \ref{thm:decreasing projection} which states that the log-concave projection of a decreasing density is decreasing.

We are not aware of precise statements on the computational complexity of the log-concave density estimation algorithms.  However, the final estimate of the log density is always a linear spline. Estimating the density with a spline is an optimization problem with constraints equal to the number of knots. We find that our final solutions typically have a small number of knots, 10 to 40, so that fitting the density is inexpensive.

\begin{theorem}
	\label{thm:decreasing density}
Let $\epsilon_i$ be i.i.d. non-negative random variables from some decreasing density or mass function $f(x)$ with support $[0, \infty)$ or the non-negative integers $\mathbb{N}$. The maximum likelihood estimator for $n$ given $V_i = n_i + \epsilon_i$ is $\hat{n} = \min_i V_i$.
\end{theorem}
\begin{proof}
This trivially follows from comparing the likelihood at $\hat{n}$ to any other point.
\end{proof}

\begin{theorem}
	\label{thm:decreasing projection}
	Let $f$ be a probability mass function with finite entropy and $\hat{f}$ be its log-concave projection. It follows that $\hat{f}$ is decreasing. 
\end{theorem}
\begin{proof}
	Given in appendix.
	\end{proof}

\subsection{Maximum likelihood estimation}
When the error density $f$ is known, the maximum likelihood estimate (MLE) for the count $n_x$ is given by  
\begin{align}
\hat{n}_x &= \arg \max_\theta \sum_{i \in \ical} \log f(V_{i} - \theta)
\end{align}
where $\ical$ is the set of counters that item $x$ hashes to.

Although the likelihood accounts for shifts in the error distribution, the maximum likelihood estimator is still often biased. However, the estimator is of the form given in section \ref{sec:debiasing}, and hence, it can be fully debiased by the bootstrap procedure in section \ref{sec:bootstrap}. 
Empirical results show this additional debiasing step is important for obtaining the best performing estimator as shown in figure \ref{fig:error}.
 Computation in this case can be moderately expensive, however, as there is no analytic form for the sampling distribution of the estimator, unlike for the Min-estimator. 

\begin{algorithm}
	\begin{algorithmic}
		\Function{Count+MLE}{$x, V$}
		\State $r \gets nrow(V)$
		\State $S \gets \{(a,h^{(a)}): i = 1,\ldots, r \}$
		\State Estimate $\hat{f}$ from $V$ using log-concave density estimation
		\State $upper \gets \min_{i \in S} V_i$
		\State $\hat{n} \gets \arg \max_{\theta \in [0,upper]} \sum_{i \in S} \log \hat{f}(V_i - \theta)$
		\EndFunction
	\end{algorithmic}
\caption{Maximum likelihood estimator}
	\label{alg:MLE}
\end{algorithm}

\subsection{Regression algorithm}
The same maximum likelihood approach can be applied for joint estimation of counts by applying linear regression with the estimated error distribution. The maximum likelihood estimate for a set of items $S$ that are hashed to indices $\ical$ is given by a maximizer of the objective
\begin{align}
\label{eqn:MLE}
\ell_f(\mathbf{\theta}_S) &=  \sum_{i \in \ical} \log f(V_{i} - M_{i, S} \mathbf{\theta}_S). 
\end{align}

When many counts are  jointly estimated, $|\ical|$ may be close to the size of the sketch. In this case, there are few counters containing purely error terms and an estimate of $f$ must also utilize information in $V_\ical$ as well. This turns equation \ref{eqn:MLE} into a joint maximization problem over both $\mathbf{\theta}_S$ and log-concave densities $f$ and requires extending the sum over previously irrelevant counters.
\begin{align}
\label{eqn:joint MLE}
\ell(\mathbf{\theta}_S, f) &=  \sum_{i} \log f(V_{i} - M_{i, S} \mathbf{\theta}_S). 
\end{align}
The maximizer for this objective is known to exist and be consistent \cite{dumbgen2010approximation}.
However, the optimization problem is only necessarily bi-convex. To estimate the maximizer, we alternate between maximizing $\mathbf{\theta}_S$ and $f$.

\begin{algorithm}
	\begin{algorithmic}
		\Function{Regression Count-Min}{$S, V$}
		\State Construct the sparse matrix $M_{\cdot, S}$ according to section \ref{sec:linear alg}
		\State Initialize $\hat{n}_S$ using, for example, the Debiased-Min estimator.
		\While{not converged}
		\State Estimate $\hat{f}$ using the log-concave density estimator on the residuals $V - M_{\cdot, S} \hat{n}_S$.
		\State $\hat{n}_S \gets \arg \min_{\theta_S} \ell_{\hat{f}}(\theta_S)\quad$ (Equation \ref{eqn:MLE})
		\EndWhile
		\EndFunction
	\end{algorithmic}
\caption{Joint count estimation using regression}
	\label{alg:regression}
\end{algorithm}

\subsection{Robust statistics}
\label{sec:robust}
When the data is heavy-tailed,  the estimated log-concave objective mirrors those used in robust statistics. 
In this case, the trimmed log-convex tail of a heavy tailed distribution is projected to a log-concave density. This results in the linear tails in the estimated error log density. These linear tails are extended so that the estimated log density has unbounded support.
Objective functions with such linear tails are robust. They are insensitive to the actual value that an outlier takes. 
For example, consider a continuously differentiable objective $\sum_i \log f(v_i - \theta)$ where the log density $\log f(v) = av+b$ on $[u, \infty)$. The maximizer satisfies $\sum_i (\log f)'(v_i - \theta) = 0$. For an outlier $v_j \gg u$, the derivative $(\log f)'(v_j - \theta) = a$ is constant for all reasonable values of $\theta \in [0, v_j - u]$. Thus, the value of $v_j$ has no effect on the solution beyond the fact that it is $ > u$. For non-differentiable objectives a similar argument applies to subgradients.

Figure \ref{fig:robust} illustrates this by showing the true log mass function for a sketch with an average of $\lambda= 64$ items hashed to each counter with item counts drawn from a $Zipf(3)$ distribution.
This is compared to the corresponding log-concave estimate for on a right trimmed sample for three different levels of trimming. 
The estimated and true distributions match well except at the log-convex right tail. In that region, the estimated distribution linearizes the tail to ensure concavity.
The trimming changes the sensitivity of the resulting objective to large counts. 
It is included since otherwise the the log-projection is not well-defined when there is a log-convex tail with unbounded support. We trim the largest $1\%$ of error values in our experiments.

\begin{figure}
	\vspace{-0.6cm}	
	\includegraphics[height=2in, width=3in]{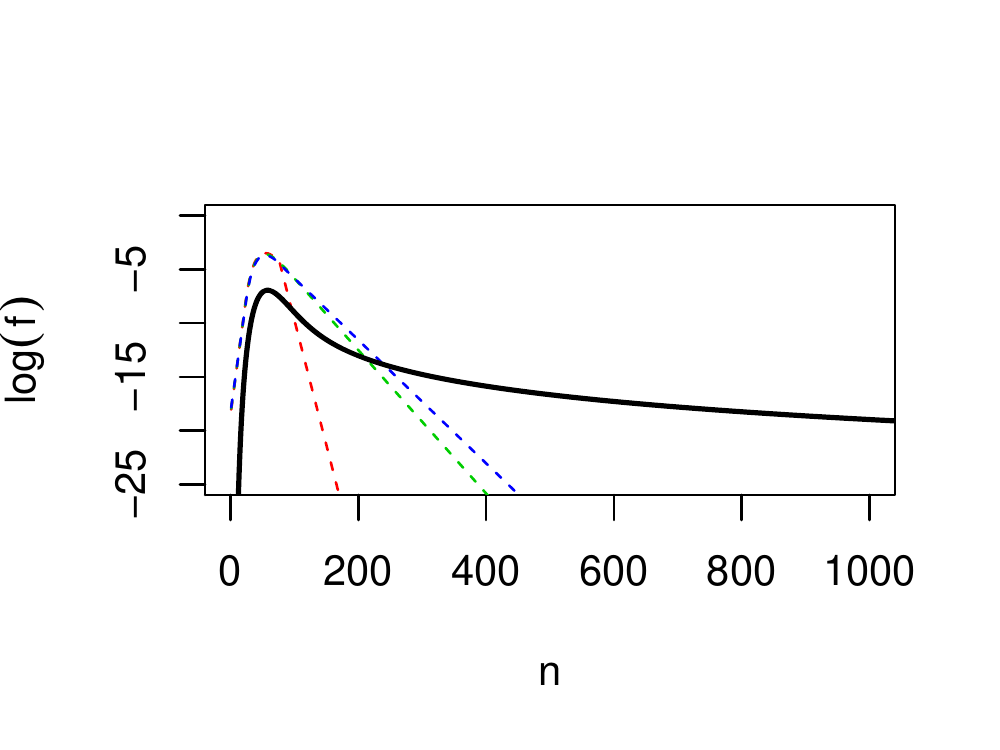}
	\vspace{-0.6cm}	
	\caption{The true mass function for a $Compound-Poisson(64, Zipf(3))$ distribution is compared to trimmed log-concave estimates. The mode and left side of the distribution match well while the right side is linearized to account for the log-convexity of the tail.}	
	\label{fig:robust}
\end{figure}

\subsection{Counter Distribution}
\label{sec:counter distribution}
We derive the exact asymptotic counter distribution given some unknown  parameters. %
The significance of this derivation is that 1) it allows one to understand when the estimation assumptions are reasonable, 2) it allows one to easily compute how the error distribution changes as sketch parameters are changed, and 3) it allows us to make precise the conditions under which our Bayesian estimator is optimal.

Under the assumption that each hash generates a completely random mapping, items are assigned to a counter $V_i$ with very small probability $1/k$. It follows from the Poisson limit theorem that the number of items in each counter is asymptotically $Poisson(\lambda)$ whenever $n/k \to \lambda$ with $k\to \infty$. 

Suppose the true counts have probability mass function $g$, and excluding item $x$, the number of items assigned to counter $V^{(a)}_i$ in replicate $a$ is denoted $W^{(a)}_i$. 
This leads to the following asymptotic observational model for a single replicate in the sketch. When $i = h^{(a)}(x)$,
\begin{align}
W^{(a)}_{i} &\sim Poisson(\lambda) \\
\epsilon^{(a)}_i &\sim g^{*W^{(a)}_i} \\
V^{(a)}_i &= n_x + \epsilon^{(a)}_i
\end{align}
where $g^{*W_b}$ denotes the convolutional power and is the distribution of the sum of $W_b$ i.i.d. $g$-distributed random variables.
The error distribution for $\epsilon^{(a)}_i$ is thus a $Compound-Poisson(\lambda, g)$.
We denote this error distribution by $F$ and its corresponding density or probability mass function by $f = g^{*W_i}$.

In general, neither $\lambda$ nor $g$ are known. 
Rather than estimating them, we directly estimate the error distribution $F$ non-parametrically under an assumption of log-concavity. Sufficient conditions for log-concavity of the error distribution are  provided by Theorem 5.5 in \cite{johnson2013log} which we restate here.

\begin{theorem}[Sufficient conditions for log-concavity]
	Let $g$ be a mass function supported on the positive integers, and $g^{\#}(x) \propto x g(x)$ be the corresponding size-biased measure. A $Compound-Poisson(\lambda, g)$ distribution is log-concave if $g^\#$ is log-concave and $\lambda \geq 2g(2) / g(1)^2 $. 
\end{theorem}

Note that $g^{\#}(x) = g(x) + \log x + constant$ so log-concavity of $g$ implies log-concavity of $g^{\#}$, and if the underlying count distribution is log-concave, then so is the error distribution for sufficiently large $\lambda$. Of particular note is the Negative-Binomial distribution which can be expressed as a compound Poisson distribution.

In this paper, the most useful property of the compound Poisson distribution is given in lemma \ref{thm:convolution} which states that the distribution resulting from increasing the rate $\lambda$ can be expressed using convolution. The resulting distribution on the interval $[0,\ell]$ can be quickly computed in $O(\ell \log \ell)$ time using a Fast-Fourier Transform. We demonstrate how this can be used to choose appropriate tuning parameters in section \ref{sec:tuning}.

\begin{lemma}
	\label{thm:convolution}
	Let $f$ be the mass function or density of a $Compound-Poisson(\lambda, g)$ distribution.
	Then, $f^{*r}$ is the mass function or density of a 
	$Compound-Poisson(r \lambda, g)$ distribution.
\end{lemma}
\begin{proof}
	This follows trivially from the superposition theorem for Poisson processes \cite{kingman1993poisson}.
\end{proof}

\subsection{Bayesian estimation}
Since our procedure produces a likelihood function, it is natural to consider the resulting Bayesian estimator given a prior. In this case, it is possible to make precise statements about the optimality of the estimator.

Given a prior distribution $\pi$ for the unknown count $n_x$ and error density $f$, the posterior distribution for $N_x$ is given by
\begin{align}
p(n_x | V_\ical) \propto \pi(n_x) \prod_{i \in \ical} f(V_i - n_x)	
\end{align}
where $\ical$ is the set of indices $x$ hashes to. 
By simply replacing $f$ with the estimated $\hat{f}$, one obtains an estimated posterior.
Given a loss $L(\theta, n_x)$, the optimal Bayesian estimator is the minimizer
\begin{align}
	\label{eqn:Bayes estimator}
\hat{n}_x &= \arg \min_{\theta} \int L(\theta, n_x) p(n_x | V_\ical) d n_x.
\end{align}

This leads to the optimality result in theorem \ref{thm:bayes optimal}. In simple terms, it states that if the number of replicates and average number of distinct items per counter stays the same but the number of error counters goes to infinity, then the Bayes optimal estimator using the approximate posterior converges to the true optimal estimator in probability.

\begin{theorem}
	\label{thm:bayes optimal}
	Let $\{N_i\}_i$ be a sequence of infinitely exchangeable counts with bounded marginal mass function $g$. Consider a sequence of Count+ summaries on the first $Poisson(d)$ counts where the sketch parameters $r$ is fixed and $k \to \infty$ as $d \to \infty$ such that $d / k \to \lambda > 0$. Let $f_\lambda$ be the mass function of a $Compound-Poisson(\lambda, g)$ and $F_\lambda$ be its c.d.f.. 
	Let $n_x^{opt}$ be the optimal Bayes estimator given in equation \ref{eqn:Bayes estimator} using a bounded loss function and $\hat{n}_x^{est}$ be the estimator using the approximate posterior obtained by estimating $f$ using the maximum likelihood log-concave density estimator and an atomic mass at $0$. Assume $f_\lambda$ is log-concave and has finite entropy. Further assume that the objective
	$J(\theta) = \int L(\theta, y) p(y | V) dy$ has a well separated maximum with probability 1. That is, given the maximizer $\theta_0$, if $J(\theta_i) \to J(\theta_0)$ then $\theta_i \to \theta_0$.  
	Then,
	\begin{align}
		n_x^{est} - n_x^{opt} \stackrel{p}{\to} 0.
	\end{align}
\end{theorem}
\begin{proof}
	Given in the appendix.
	\end{proof}

We note that this optimality result is a strong finite sample result, as only $r$ counters contain an item's count,
 rather than an asymptotic optimality result or an even weaker rate result that is typical in the literature. Only finitely many replicates are observed for each item of interest. 

\section{Asymptotics}
\label{sec:asymptotics}
There is a rich set of work on asymptotics that aids understanding what  makes a count estimator statistically efficient.
The more general form of the estimation problem is to find the true count $n_x$ from a set of observations drawn from a density $f(y - n_x)$ where the error density $f$ has support in $[0, \infty)$.
Such a problem is sometimes referred to as endpoint estimation.
A number of works including \cite{hall1982endpoint}, \cite{woodroofe1974}, \cite{cooke1979statistical} focus on the difficult case where the density vanishes at 0.
More specifically, they consider the case where the density is of the form $f(y) = (y/c)^\alpha (1+O(x))$. In this case, one finds 3 regimes. When $0 < \alpha <1$, so that the density drops off sharply near 0, the Min-estimator is nearly optimal. 
When $\alpha > 1$ so that the density decays slowly near 0, the best possible rate using only points close to the minimum is $o(r^{-1/2})$ and worse than the $\Theta(r^{-1/2})$ rate achieved by the mean. In this case the support constraint provides little value.
When $\alpha = 1$, one can achieve the improved rate $O((r \log \tilde{r})^{-1/2})$
using maximum likelihood estimation on $\tilde{r} = o(r^{2/3})$ items closest to the minimum.

These three regimes are illustrated in figure \ref{fig:truncations} which show the behavior of various estimators under different truncations of a normal distribution. Truncations which are far into the tail of the normal distribution have that the Mean estimator is optimal. Truncations near the mode find that a debiased Min-estimator is near optimal. Truncations in the middle find both the Mean and Min-estimators deviate from the optimal estimator.

\begin{figure}
	\begin{tabular}{cc}
	\includegraphics[height=1.5in, width=1.5in]{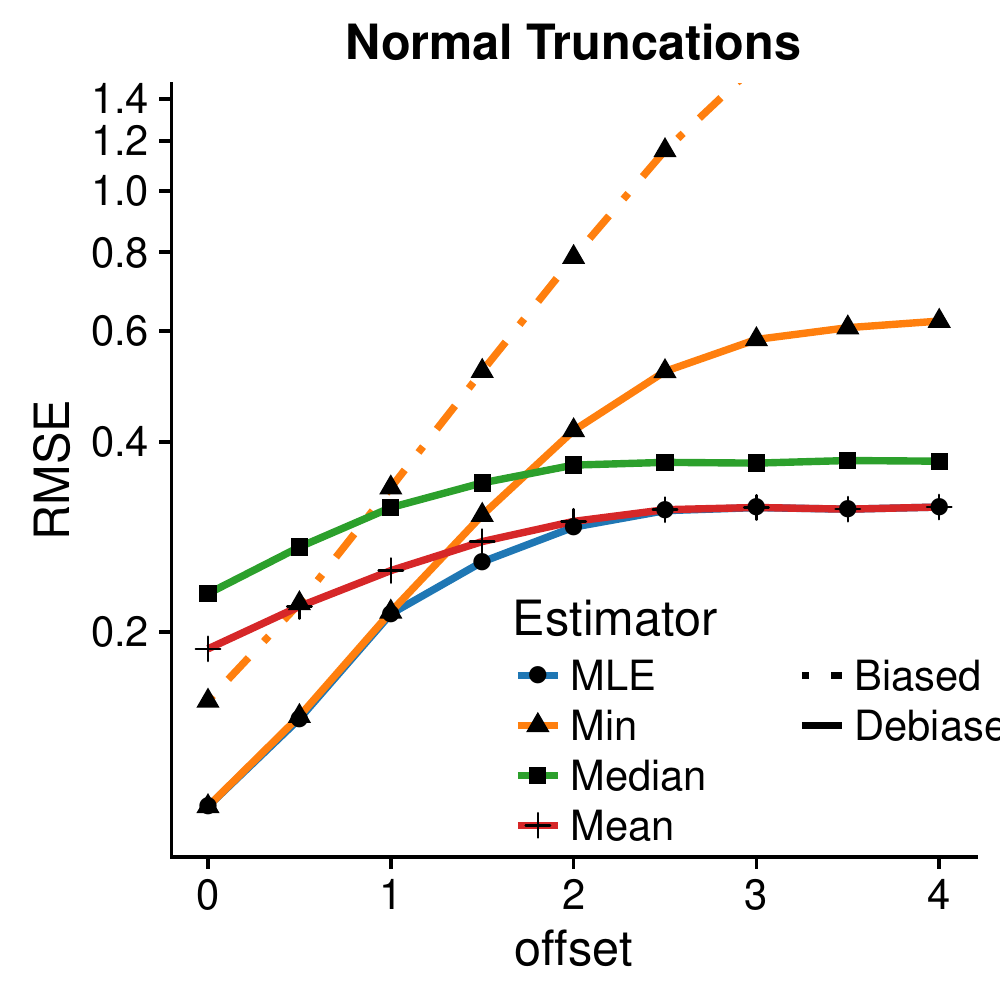} &
	\includegraphics[height=1.5in, width=1.5in]{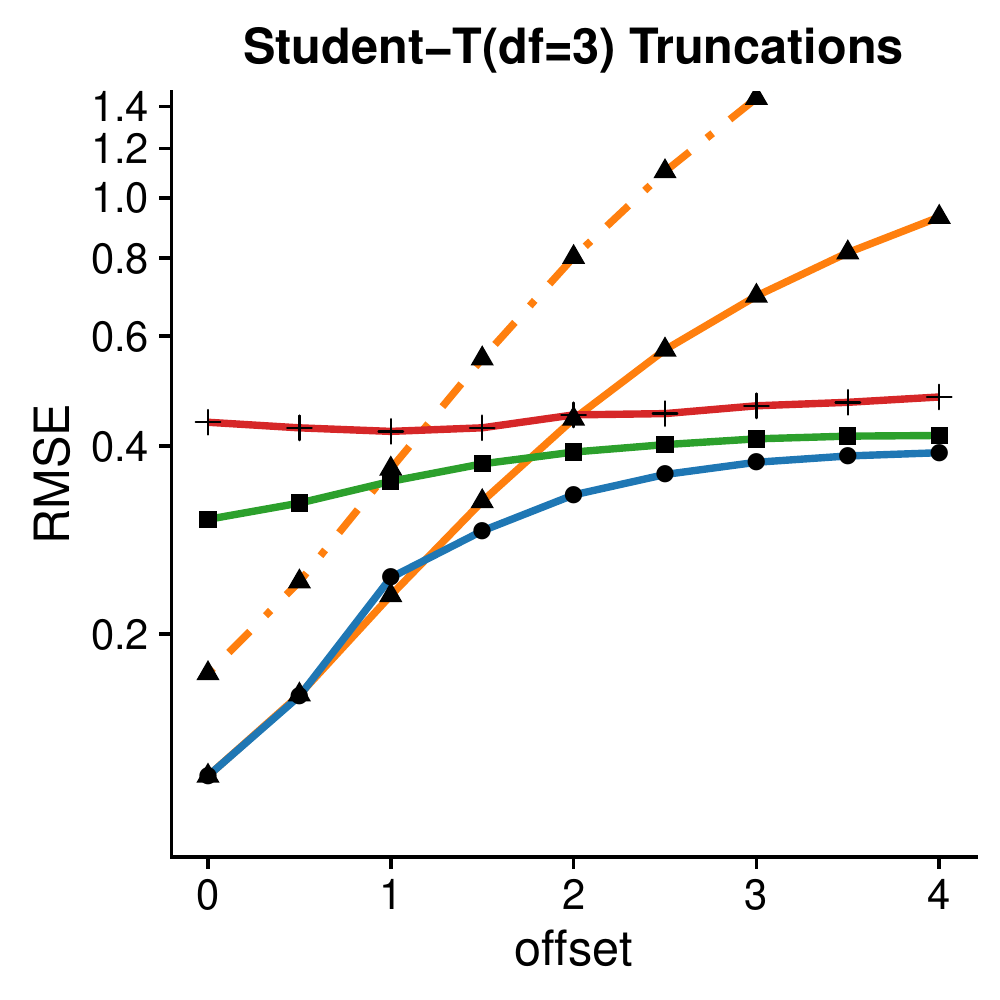}
	\end{tabular}
\vspace{-0.1cm}
\includegraphics[height=0.75in, width=3in]{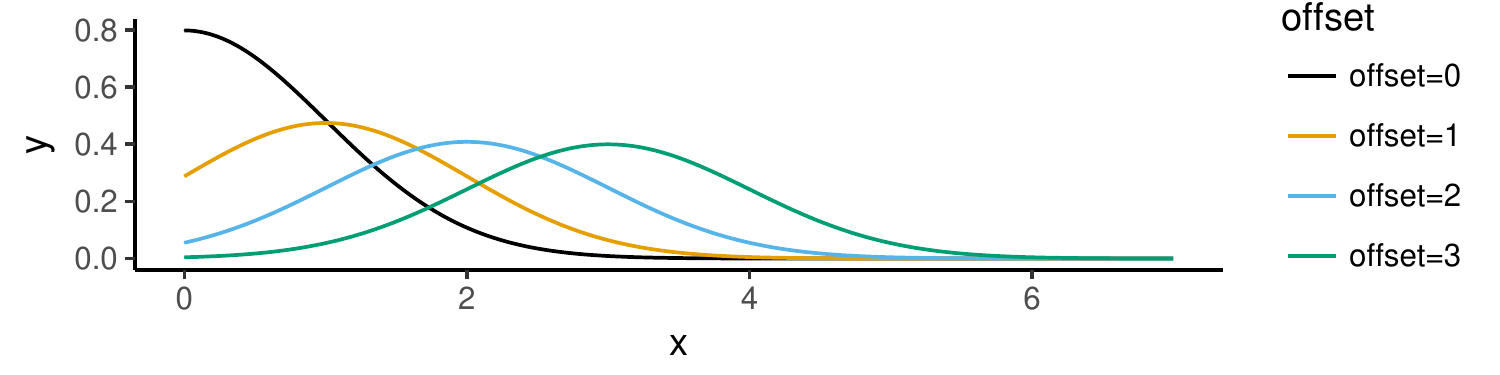}
\vspace{-0.2cm}
\caption{Performance of Debiased estimators under different error distribution shapes. The MLE is always the best or nearly the best estimator. For the normal distribution, the Min estimator is optimal in the regime where the density near 0 is high (lower figure), and likewise the Mean is optimal when the density low. For the t-distribution, none of the simple statistics, including the median, are optimal for larger offsets, and the MLE is able to beat all of them. }
	\label{fig:truncations}
\end{figure}

\section{Tuning sketch parameters}
	\label{sec:tuning}
	Although our methods take the guesswork out of what estimation procedure to choose, the sketch creator must still choose the number of replicates $r$ and the number of counters per replicate $k$, or width. The original Count-Min paper \cite{cormode2005countmin} suggests choosing these to minimize the space required to achieve a desired error guarantee.
	For the guarantee, $P(\hat{n}_x \leq n_x + \epsilon \| \mathbf{n} \|_1) < \delta$, their error bound yields the suggestion $r = \lceil \log (1/\delta) \rceil$ and $m = \lceil e / \epsilon \rceil$. It has been suggested \cite{cormode2012synopses} that typically $r \approx 10-30$ in practice but can be as low as $4$ \cite{cormode2008finding} without obvious ill-effects.  Several industry implementations 
	such as the RedisLabs module \cite{rediscountmin} choose a default of $r=10$.
	
	The previous suggestion finds the smallest sketch that will guarantee a certain confidence level and interval width based on a loose confidence bound. The same can be applied to our tight confidence intervals.
	We demonstrate how this can be done efficiently without trial and error
	by using the counter distribution from section \ref{sec:counter distribution}.

	We first consider the natural case where there is a fixed memory budget $B = rk$, and one desires the smallest interval width. As the asymptotic theory suggests the region where the Min estimator is optimal or near optimal is the best regime, it is sensible to minimize the width of the Min estimator's interval. 
	Let $F_\lambda$ be the distribution function of a $Compound-Poisson(\lambda, g)$ distribution where $g$ is the distribution of item counts, and $Beta_r$ be the distribution function of a $Beta(1, r)$ random variable. 
Given a desired confidence level $\ell$ for the one-sided confidence interval, the choice of $r$ is
\begin{align}
\label{eqn:optimal tradeoff}
\hat{r}_\ell = \arg \min_\rho F_{d \cdot r / B}^{-1} (Beta_\rho^{-1}(\ell)).
\end{align}

This is easily computed from a single  $1 \times B$ Count+ summary.
The summary provides the error distribution $F_{\lambda_0}$ where the rate $\lambda_0 = d/B$ and a corresponding density estimate of $f_{\lambda_0}$. Lemma \ref{thm:convolution} gives that the error distribution for any choice of parameters $r \times k$ can be computed
as the convolutional power $f_{\lambda_0}^{*r}$,
which can be efficiently computed using a Fast-Fourier transform. 
Figure \ref{fig:parameter tuning} illustrates how the interval width changes with $r$ for a range of confidence levels and fixed memory budget. 

Furthermore, the underlying data can be downsampled using coordinated or bottom-k sampling  \cite{cohen2013coordinated} to estimate error distributions with even smaller rates. This allows one to explore the confidence interval widths for a range of sketch sizes as well. 

As an illustration of how this can be applied in a database system, consider the Google N-gram viewer which deals with the canonical natural language processing task of computing counts of n-grams. An n-gram is of a sequence of $n$ words. For example, "An n-gram consists" is  a 3-gram. The number of n-grams and possible pointwise queries is very large. One study \cite{yang2007ngram} found there were on the order of $10^{10}$ unique 5-grams in 100 million English web pages out of which $\approx 10^9$ appeared at least 5 times. Naively tuning parameters is costly. It requires computing a large number of exact counts as well as repeatedly  computing a sketch and estimated counts for a large number of parameter settings.  Our method shows that no true counts need to be computed, the error is obtained by a single quantile calculation, and only one sketch needs to be computed for all parameter settings.

Even when prior information about the error or count distribution is unavailable, the asymptotic theory provides guidance on how to choose the sketch parameters as wider sketches tend to be closer to the "super-efficient" regime where the Min estimator is nearly optimal.

\begin{figure}
	\begin{tabular}{cc}
	\hspace{-0.5cm}\includegraphics[width=1.75in, height=1.5in]{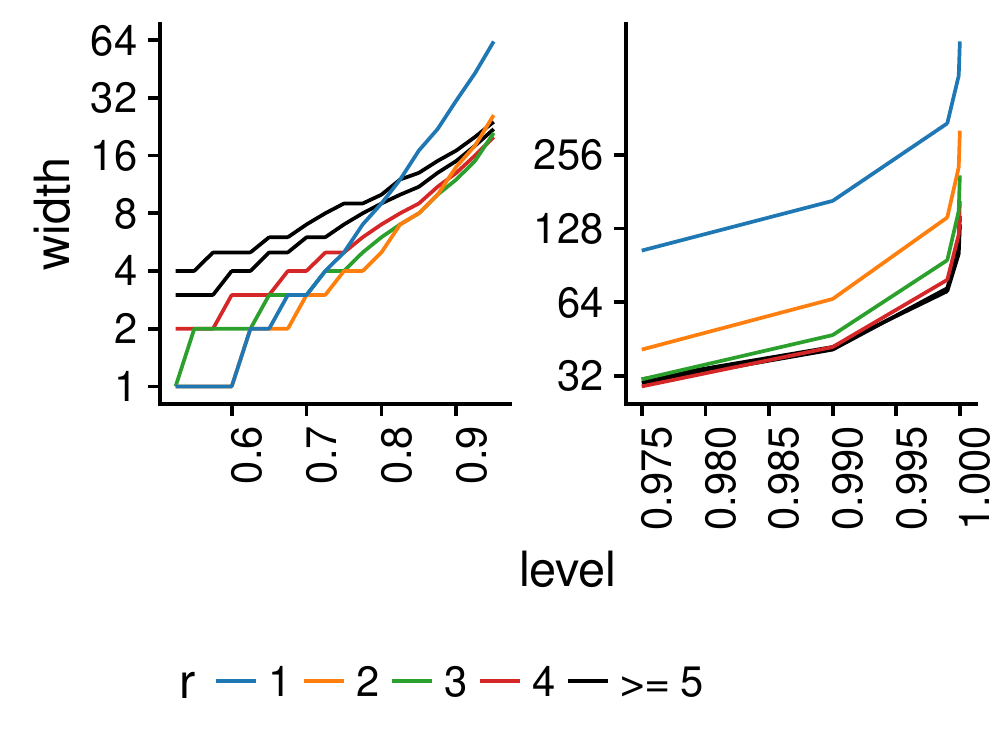} &
	\hspace{-0.2cm}\includegraphics[width=1.75in, height=1.5in]{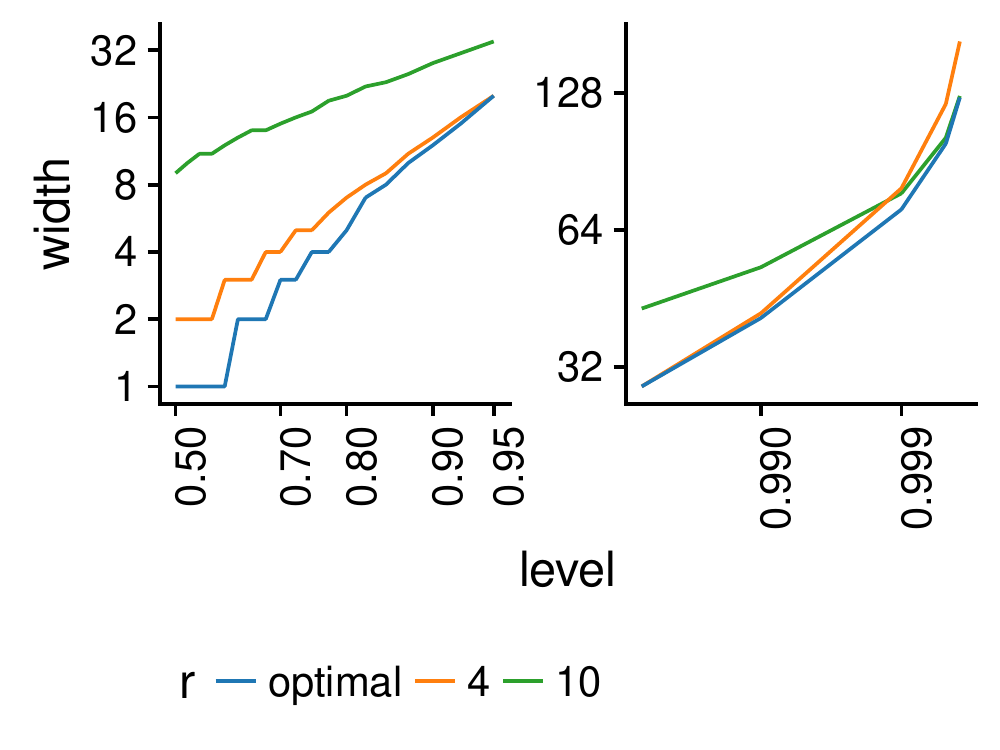}
	\end{tabular}
\vspace{-0.2cm}
\caption{Confidence interval widths for different sketch tuning parameters given a $Negative-Binomial(30, 0.01)$ count distribution and fixed memory budget. Left: Shallow, wide sketches outperform deep sketches except at high confidence levels. Right: Two previously suggested depth settings are compared to the optimized one for each level. Optimized parameters can yield much narrower confidence intervals.}
	\label{fig:parameter tuning}
\end{figure}

\section{Empirical Results}
\label{sec:empirical}
We test our MLE estimator in a variety of real and synthetic situations.
It is shown to match or best other estimators in all situations.
We also empirically show that our confidence intervals provide the correct coverage. A comparison of these tight bounds with prior bounds shows that they are orders of magnitude better.

For synthetic simulations, we use the family of Zipf-Mandelbrot, or discrete power law, distributions. These distributions have probability mass function given by $p(x) \propto (a+x)^{-\alpha}$ on the positive integers.
Here $a$ is some offset that adjusts the mass near $1$ with smaller values having a larger mass at 1, and $\alpha$ controls the tail behavior with smaller values having heavier tails. For $\alpha = 2$, the distribution has infinite variance. We always consider a universe with $d=10^6$ items.

For real world datasets, we used a network and a natural language processing dataset.
For network data, we used the CAIDA Anonymized OC48 Internet Traces dataset \cite{oc48}. In 15 minutes of network traffic there were 21.8 million packets from 1.6 million distinct source addresses and ports. We use a Count+ summary to estimate the number of packets for each source.
For natural language processing data, we used the Google N-grams dataset \cite{michel2011quantitative} for all 2-grams starting with the letters 'ta'. There are 1.4 million distinct 2-grams out of a total of 713 million.

We used the R package {\bf logcondens} \cite{dumbgen2010logcondens} to perform log-concave density estimation though we note there is a corresponding package {\bf logcondiscr} \cite{balabdaoui2013asymptotics} for discrete distribution. Although our data is discrete, we chose the continuous valued density estimation package so that resulting objective function is continuous and can be easily solved by a standard real-valued optimizer.

Although we do not consider timings for our simulation to be representative for practical implementation as R is slow, we report that count estimation for 2000 counts for a sketch of size $8 \times 10^6$ took roughly 4 ms per count on a 2.4Ghz CPU when running on a single thread. On average, each count estimate used roughly 16 evaluations of the objective function when using the function {\bf optimize} which does not make use of known gradient or Hessian information.

To compare the sketches, we use the root mean squared error and the relative efficiency. The relative efficiency of estimator $\phi_1$ to $\phi_2$ on random data $X$ is
\begin{align}
RelativeEfficiency(\phi_1, \phi_2) = \frac{\E\|\phi_2(X) - \theta\|_2^2}{\E\|\phi_1(X) - \theta\|_2^2}
\end{align}
where $\theta$ are the true values being estimated.
For unbiased estimators of real valued $\theta$ this computes the ratio of the variances, and under regular assumptions where the variance scales inversely to sample size, the relative efficiency of $\beta$ represents needing $\beta$ times more data for estimator $\phi_1$ to achieve the same error as $\phi_2$ 

For single count estimation, we compare the following estimators: the Min,  Debiased Min,  Debiased Mean,  Debiased Median,  MLE, and Debiased MLE estimators. 
Of these, the MLE estimators are the only completely new estimators. 
Other estimators benefit from our computational simplification when applicable. For all these estimators, the tight confidence intervals are from our new bootstrap procedure.
For each sketch, we estimate the counts for the top 2000 heavy hitters. In simulations, the sketch sizes range in depth from 2 to 16 replicates and width from $10^4$ to $5\times10^5$ counters per replicate. 
Figure \ref{fig:error} shows the empirical error and efficiency under the real and synthetic scenarios. The debiased MLE estimator is clearly the best estimator under all scenarios. 

\begin{figure*}
	\begin{multicols}{2}
\begin{tabular}{cc}
\hspace{-0.2cm}	\includegraphics[width=3.4in, height=5in]{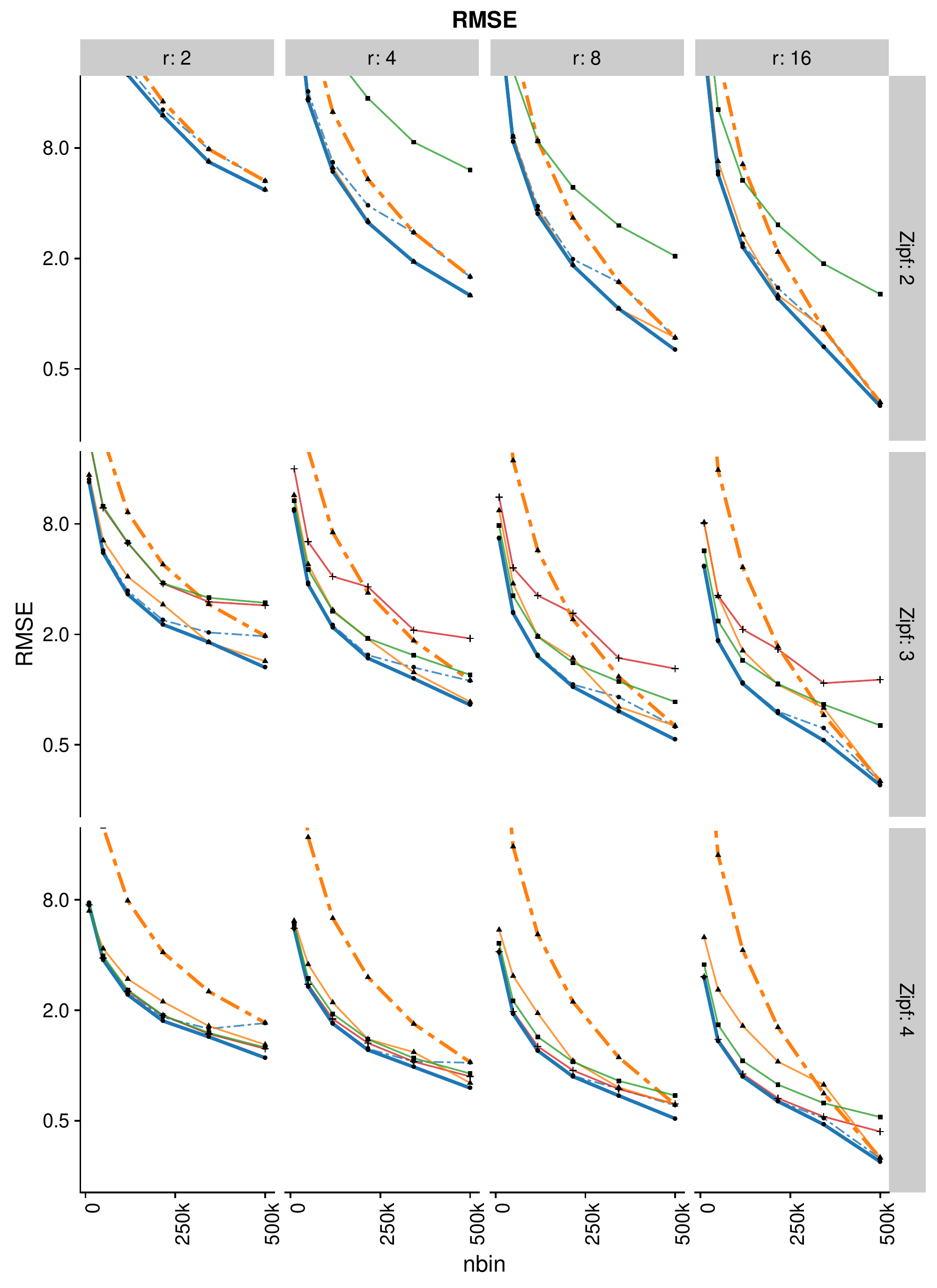} 
&
	\includegraphics[width=3.4in, height=5in]{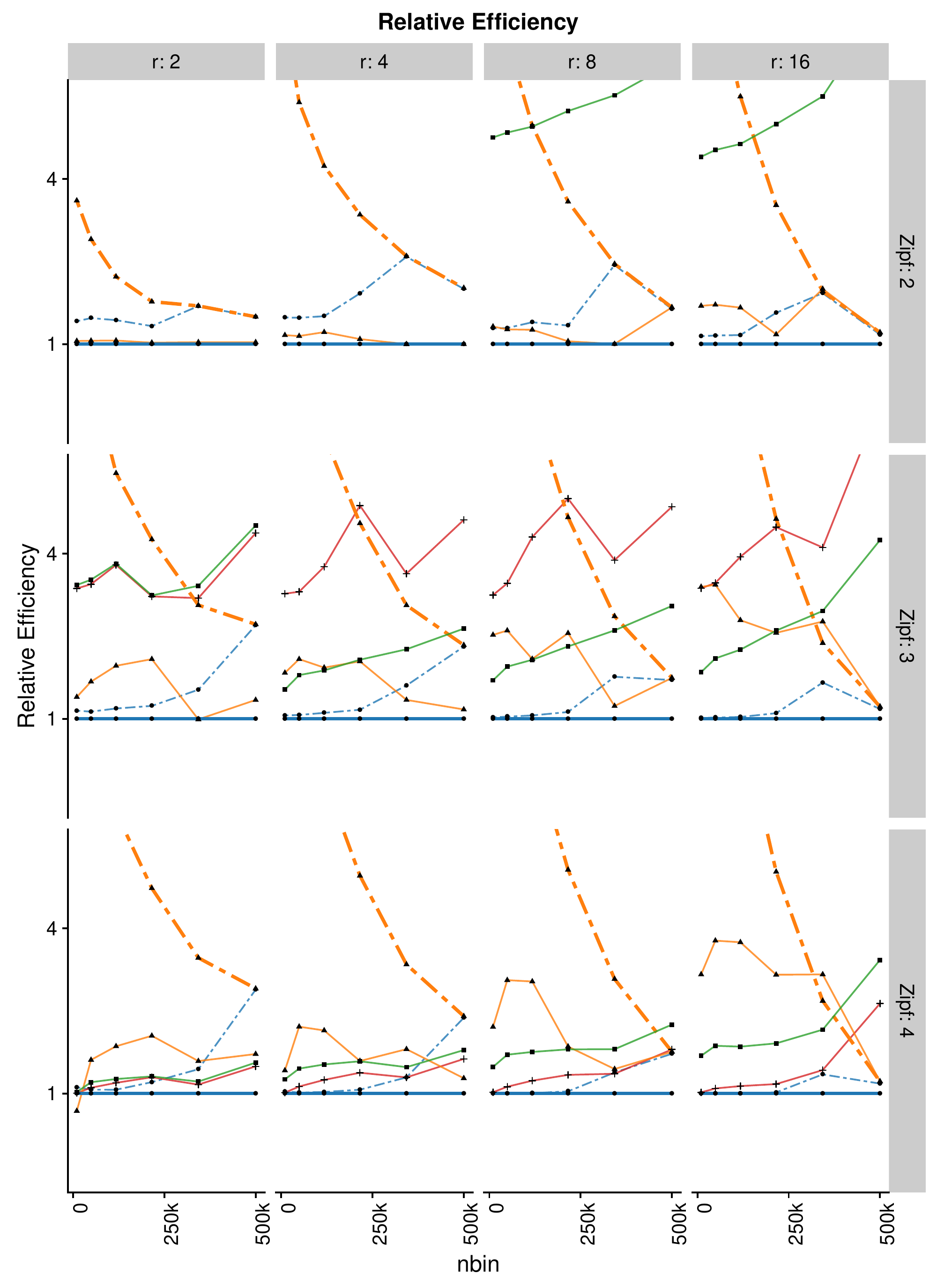} 
\end{tabular}
\begin{tabular}{ccccc}
\hspace{-0.2cm}	\includegraphics[width=1.4in, height=1.75in]{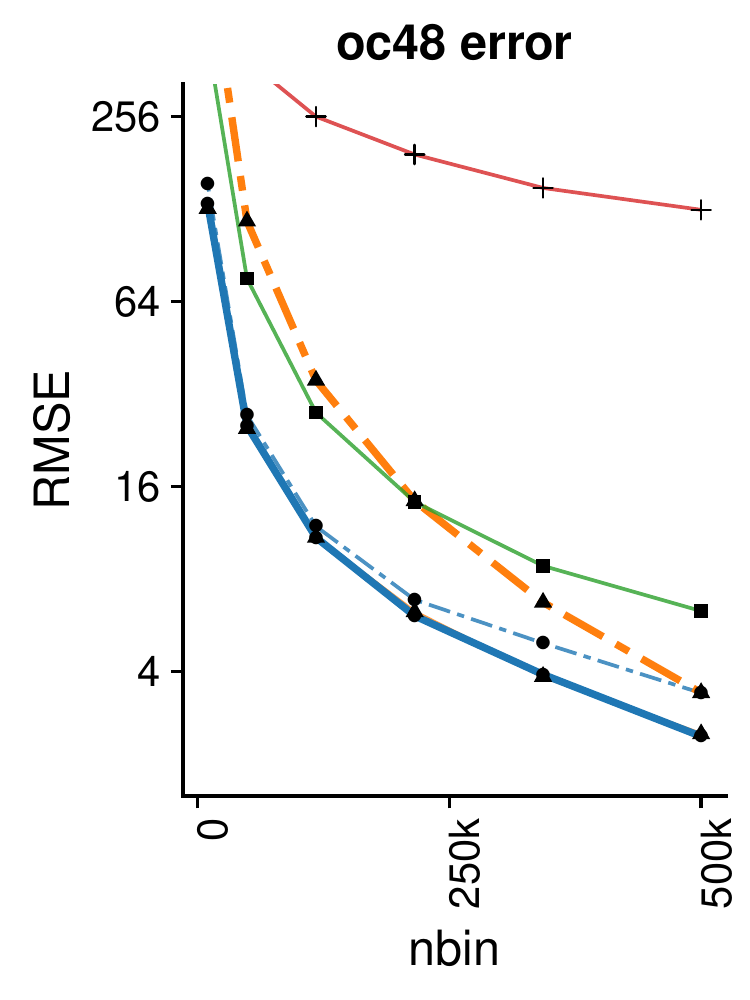} &
	\includegraphics[width=1.4in, height=1.75in]{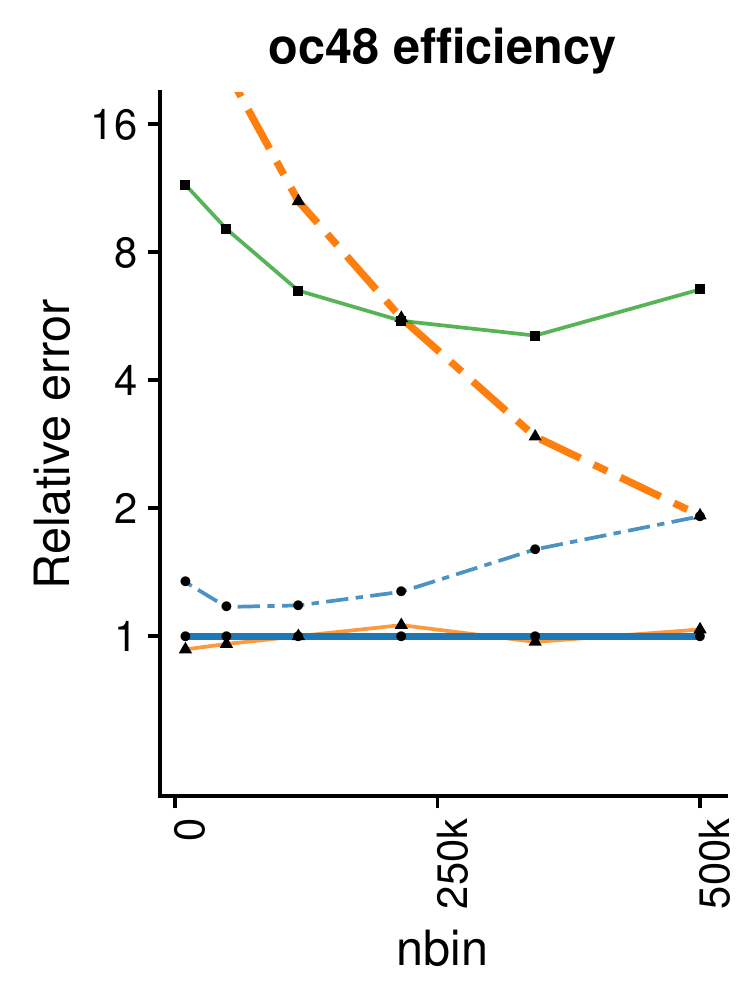} & 
\hspace{-0.2cm}	\includegraphics[width=1.6in, height=1.75in]{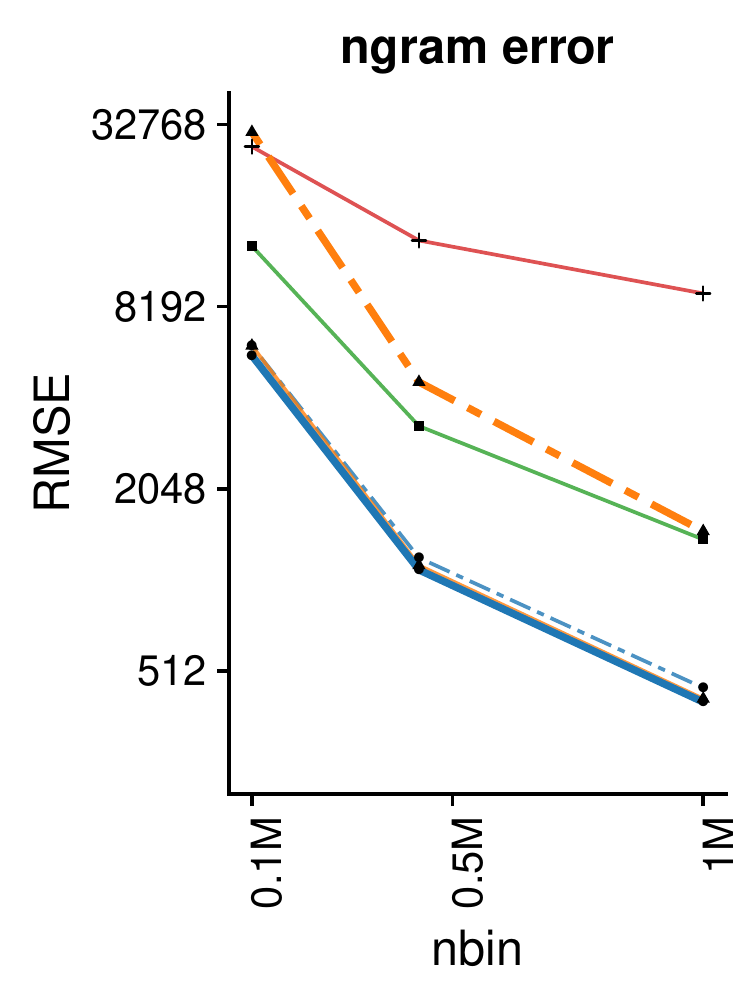} &
	\includegraphics[width=1.4in, height=1.75in]{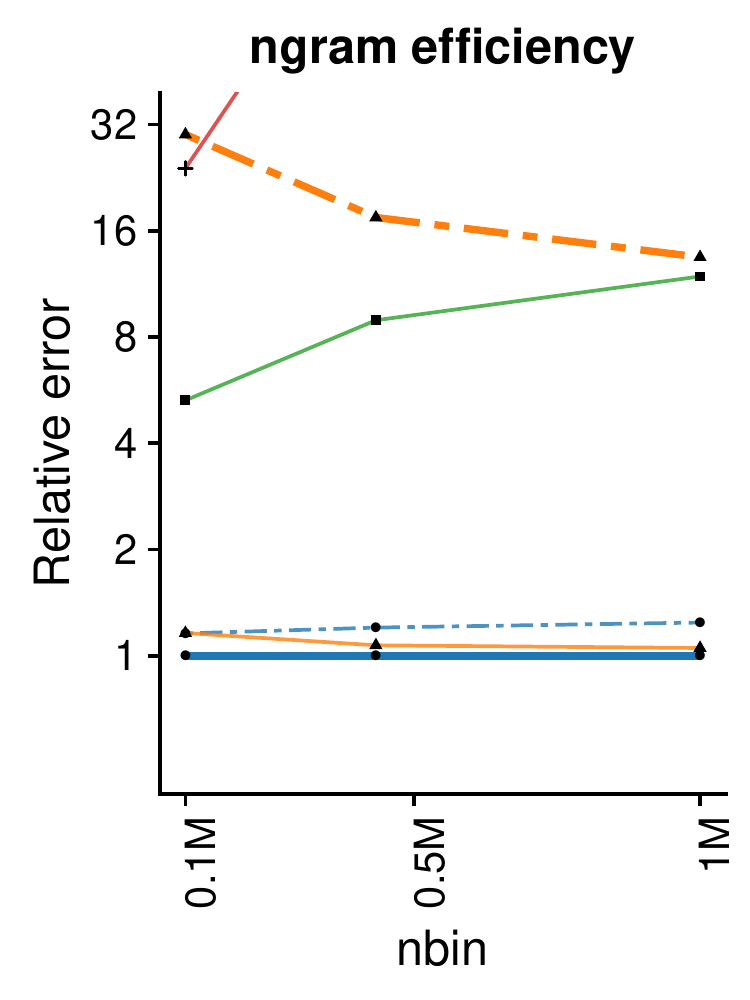} &
	\includegraphics[width=1.2in, height=1.75in]{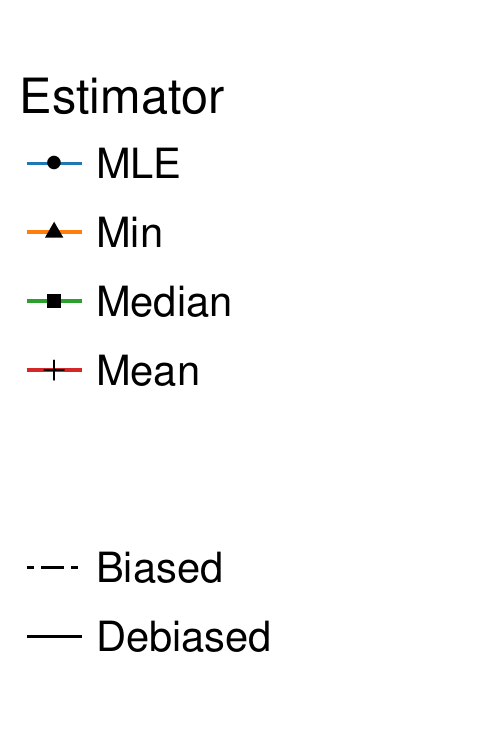}
	\end{tabular}
	\end{multicols}
\caption{The figures on the top show the performance of different estimators over a range of distribution skews and sketch parameters while the bottom figures are on real world datasets.
	The Debiased MLE estimator is the most accurate estimator in all scenarios. The Debiased Min estimator is competitive when there are heavy tails and particularly in the real datasets, but the basic Count-Min estimator (orange dashed) is significantly worse than the Debiased MLE (solid blue) estimator. For the n-gram data, the relative efficiency of the Debiased Min estimator is 1.1 times worse than the Debiased MLE estimator. Although there is no perceptible difference on the oc48 data, the difference becomes more pronounced when counting distinct ip flows (Figure \ref{fig:oc48distinct} in the appendix).
	Estimators that do not appear on the plots (such as the Mean estimator on the $Zipf(2)$ data) have very large errors that do not fit within the axes.} 
	\label{fig:error}
\end{figure*}

Figure \ref{fig:CI widths} shows the coverage of the corresponding confidence intervals for each of the estimators. They match the desired confidence levels at all levels in a multitude of settings.
The resulting error bounds are orders of magnitude better than those available from theoretical analysis.

\subsection{Regression results}
For our regression simulation, we use a small sketch with width $1024$ counters like \cite{lee2005improving} and depth 4. The count distribution is $Zipf(2)$. 
We note that this choice of small sketch is to highlight the limited regime in which regression will significantly improve the count estimator, namely the regime in which items of interest have hash collisions with the known heavy hitters. As shown in \cite{lee2005improving} and figure \ref{fig:regression}, regression yields almost no effect when the sketch is wide relative to the number of items of interest. When collisions with other heavy hitters are likely, then the error can be substantially reduced.

\begin{figure}
	\includegraphics[height=1.75in, width=3.4in]{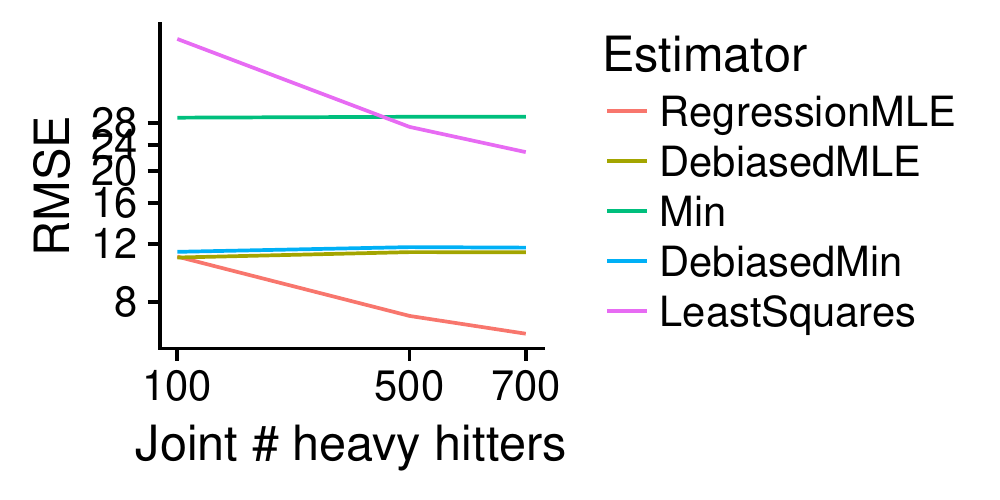}
	\caption{The RMSE per item decreases as more heavy hitter counts are jointly estimated for the regression techniques. A poor least squares objective for a $Zipf(2)$ count distribution yields inaccurate estimates that are always worse than the Debiased MLE and Min estimators despite taking advantage of additional information that can account for over 45\% of the total count. For small numbers of items being jointly estimated, regression based techniques are not useful.}
	\label{fig:regression}
\end{figure}

 \section{Discussion}
 We discuss the applicability of our techniques to other counting sketches and inner product estimation and address computational issues that arise with using empirical error distributions and likelihood based estimators.
 
\subsection{Application to other counting sketches}
The same idea of empirically estimating an error distribution to improve count estimation can be applied to other counting sketches and modifications of the Count+ summary. It is straightforward to apply to linear sketches such as the Count sketch \cite{charikar2002countsketch} and modifications to the Count+ summary that preserve linearity, for example the time adaptive Ada-sketch \cite{shrivastava2016time}. 

We note, however, that there is little reason to prefer the Count summary over the Count+ summary for pointwise queries when using our likelihood based estimators. 
Prior results \cite{deng2007new} show that the accuracy of the Debiased Median estimates closely matches the Count sketch estimates.
The Count summary is the same as the Count+ summary except item $x$'s counter is randomly incremented by either $-n_x$ or $n_x$ rather than $n_x$. Thus, the error terms are necessarily more noisy than those in a Count+ summary, and estimation should not be expected to be better when exploiting the full likelihood.

Several other modifications can be described as random non-linear transformations of an underlying linear Count+ summary.
These include the Count-Min-Log sketch \cite{talbot2009countminlog}, \cite{pitel2015count} which uses approximate counter to save space
and a proposed sketch to replace the simple additive counters with approximate distinct counters \cite{cormode2009forward} .
The same idea of empirically estimating an error distribution and applying statistical estimation techniques can be applied to improve estimation. The resulting estimators are more complex as  the observed counter values cannot be used directly. Computing the likelihood requires integrating over the error distribution.

Non-linear sketches such as  the Conservative Update Count-Min sketch
 result in summaries where the error terms are no longer exchangeable. The irrelevant counters for an item are not necessarily informative of the error distribution in the relevant counters. The conservative update modification updates only the smallest of the $r$ counters that an item hashes. This substantially reduces the raw magnitude of the error vector and potentially improves performance in the regime where the biased Min-estimator is nearly optimal.
  However, in other regimes, the error will still grow linearly as $O(\lambda)$
  since there is no debiasing operation. In contrast, the error in the Count+MLE estimator will grow with the standard deviation $O(\sqrt{\lambda})$. Furthermore, there is no procedure for generating tight confidence intervals when using conservative updates.

\subsection{Computational complexity and Application to streaming settings}
Thus far, estimation of the empirical error distribution  has been assumed to have manageable computational cost. This is aided by the fact that if a sketch does not change, then the error distribution only needs to be estimated once. This may not be the case in streaming settings. Furthermore, in extremely high throughput situations, the maximum likelihood estimator may also be relatively expensive to compute in comparison to simple estimators like the Min, Mean, and Median. 

These problems may be alleviated in two ways. First, the estimated error distribution can be updated infrequently. If the empirical distribution is updated only when it can differ by $\delta$ so that $\| \mathbb{F}_n - \hat{F}_{current}\|_\infty < \delta$, then the number of times the estimated error distribution is updated is logarithmic in the stream size.
The amortized cost of adding a count to the sketch goes to 0. 
Second, rather than using the MLE estimator, the tight error bounds can be used to periodically select the best simple estimator. Thus, the estimator can smoothly transition from the regime where the Min estimator is optimal to ones where the Mean or some quantile estimator is optimal.

	\subsection{Inner products}
	Inner product estimation is somewhat more challenging than item count estimation. All counters contain relevant item counts
	so the error distribution cannot be simply gleaned from unused counters.
	We provide a means to generate an approximate error distribution but do not evaluate this procedure in this paper as we regard it as a substantial separate topic.
	
	Given two sketches $\mathbf{U}, \mathbf{V}$ for true count vectors $\mathbf{m}, \mathbf{n}$, the naive Count-Min inner product estimate for $\mathbf{m}^T \mathbf{n}$ is $\min_a U^{(a)T}V^{(a)}$. It has been empirically shown to perform well when the data is highly skewed \cite{rusu2007statistical}, but in cases where there is a heavy tail, the bias in the estimate is large, and it can perform an order of magnitude worse than other methods. 
	An unbiased estimator is given by \cite{thorup2004tabulation} which empirically performs similarly to the AMS sketch \cite{alon2005estimating}. The AMS sketch is the same as the Count sketch but with an inner-product estimator instead of a item count estimator. 
	
	To find an error distribution, expand the product of two counters to identify the form of the error. Let $i$ be some index and $S_i$ be the set of all items that hash to that index. This error in the product is given by
	\begin{align}
	U_i^{(a)T} V_i^{(a)} &= \sum_{x \in S_i} m_x n_x + \epsilon_i^{(a)} \\
	\epsilon_i^{(a)}  &= \sum_{x,y \in S_i, x \neq y} m_x n_y.
	\end{align}
	The error is the sum of $|S_i|^2$ pairs of counts where  $|S_i| \sim Poisson(\lambda)$ and for each item pair $(X,Y)$, $X$ and $Y$ are drawn independently, though there is dependence between pairs.
	An imperfect surrogate of the error distribution can be obtained by multiplying random counters in the sketch. For indices $i,j$,
	\begin{align}
	U_i^{(a)} V_i^{(a)} = \sum_{x \in S_i, y \in S_j} m_x n_y.
	\end{align}
	This ensure the number of pairs is approximately correct when $\lambda$ is large and preserves part of the dependence structure between pairs.
	Rather than explicitly constructing a sample, this error distribution can be computed by estimating a distribution for log counter values and taking the convolution.

\section{Conclusion}
This paper addresses a number of practical problems for counting sketches and advances our understanding of the mechanisms by which they work. We provide two distinct primary contributions. 1) We give the first method that produces practical and tight error estimates for a pointwise query, and 2) we derive improved and optimal estimators that make full use of the information contained in the sketch. 
	Besides their immediate contributions to counting sketches, we show they help solve other problems facing a practitioner including which sketch and which count estimator to use and how to select optimal sketch tuning parameters.  
	
	\bibliographystyle{abbrv}
	\bibliography{ling2}
	
	\appendix
	\section{Proofs}
	
	Theorem \ref{thm:bayes optimal}
	\begin{proof}
		Since $d/k \to \lambda$, it follows that the error distribution $\|f - f_\lambda\|_\infty \to 0$.
		By Theorem 3.4 in \cite{balabdaoui2013asymptotics},
		$\| \hat{f} - f_{\lambda} \|_\infty \to 0$.
		For any $\epsilon >0$, choose $u_\epsilon$ such that $F_\lambda(u_\epsilon) > 1- \epsilon$, and take $M = \max_y f_\lambda(y)$.
		Hence, for any $\delta > 0$, 
		the approximate likelihood $\left\| \prod_{i=1}^r \hat{f}(V_i - \theta) - \prod_{i=1}^r f(V_i -\theta)  \right\|_{\infty} < \delta + p(\max_{i=1, \ldots, r} \{V_i\} > u_\epsilon) < \delta + (M + \delta)^r (1-(1-\epsilon)^r)$ with probability $< 2 \epsilon$ eventually.
		Thus, the approximate posterior uniformly converges to the true posterior in probability, $\| \hat{p}(y | V) - p(y | V)\|_{\infty} \stackrel{p}{\to} 0$. 	
		Boundedness of the loss function ensures uniform convergence of the objectives 
		$\left\| \hat{J}(\theta) - J(\theta) \right\|_\infty \stackrel{p}{\to} 0$ where
		and $\hat{J}$ is similarly defined on the approximate posterior. 
		The well-separation gives the desired convergence in probability of the maximizers by the M-estimation consistency theorem \cite{vandervaart}.
	\end{proof}
	
	Theorem \ref{thm:decreasing projection}
	\begin{proof}
		The log-concave projection $\hat{f}$ is the maximizer of $J(g) = \int f \log g$ over log-concave mass functions.
		Assume $\hat{f}$ is not decreasing. Without loss of generality assume, the left endpoint of the support of $f$ is 0.
		Let $i$ be the smallest value such that $\hat{f}(i)^2 > \hat{f}(i-1)\hat{f}(i+1)$.	
		Such a value must exist since otherwise $\log \hat{f}$ is linearly increasing with bounded support. Since the uniform distribution is log-concave and attains a higher objective value, $\log \hat{f}$ cannot be linearly increasing.
		
		For $\delta \in (0,1)$, define $\tilde{f}(x) = \hat{f}(x) (1-\delta)$ if $x = i$, $\hat{f}(x)$ if $ x > i$, and $\hat{f}(i)(1+\beta)$ if $x < i$ where $\beta = \hat{f}(i) / \hat{F}(i-1) \delta$. It is easy to verify that $\tilde{f}$ is a probability mass function, and that it is log-concave on $[0, i-1]$ and $[i+1, \infty)$. We will verify that it satisfies the condition for log-concavity at $i$ for sufficiently small $\delta$. 
		$2\log \tilde{f}(i) = 2\log \hat{f} + 2\log(1-\delta) > 
		\log \hat{f}(i-1) + \log \hat{f}(i+1) + 2\log(1-\delta) = 
		\log \tilde{f}(i-1) + \log \tilde{f}(i+1) - \log (1+\beta) + 2\log(1-\delta)
		= \log \tilde{f}(i-1) + \log \tilde{f}(i+1) - O(\delta)$.
		It follows that for small enough $\delta$, the condition 
		$2\log \tilde{f}(i) \geq \log \tilde{f}(i-1) + \log \tilde{f}(i+1)$ holds.
		We can now show that $\tilde{f}$ also attains a higher objective value for small enough $\delta$.
		$J(\tilde{f}) - J(\hat{f}) = \sum_x f(x) (\log \tilde{f}(x) - \log \hat{f}(x)) =
		\sum_{x < i} f(x) \log (1+\beta) + f(i) \log(1-\delta) = F(i-1) \log(1+\beta) + f(i) (1-\delta) = F(i-1) \frac{\hat{f}(i)}{\hat{F}(i-1)} - f(i) \delta + o(\delta)$.
		Since $f$ is decreasing and $\hat{f}$ is stricty increasing on $[0,i]$,
		$f(i) \leq F(i-1) / (i-1)$ and $\hat{f}(i) > \hat{F}(i-1) / (i-1)$. Taking $\epsilon = \hat{f}(i) / \hat{F}(i-1) - 1/(i-1)$,
		it follows
		that $J(\tilde{f}) - J(\hat{f}) > F(i-1) \epsilon  + o(\delta)$ which is $> 0$ for sufficiently small $\delta$. 
		Thus any non-decreasing $\hat{f}$ can not be a maximizer.
	\end{proof}

\subsection{Counter Braids}
\label{sec:appendix counter braids}
When the entire universe of items is known, the true counts must be in the feasible set $\Theta = \{\theta >0 : \mathbf{M} \theta = \mathbf{V} \}$. This set is expensive to compute when the number of counters is large. The counter braids estimator instead keeps track of upper and lower bounds for the feasible set so that 
\begin{align}
\tilde{n} &\leq  \mathbf{n} \leq \hat{n}.  \label{eqn:cb 1}
\end{align}
It follows that
\begin{align}
0 \leq V - M \tilde{n} &= M (\mathbf{n} - \tilde{n}) \label{eqn:cb 2}\\ 
0 \leq M \hat{n}  - V &= M (\hat{n} - \mathbf{n}) \label{eqn:cb 5}
\end{align}
which has the same form as a Count+ estimation problem. Using the Min estimator, which only exploits the non-negative support constraint, yields the updates
\begin{align}
\hat{n}_x &= \min_i (V_i(x) - M_{(i,h^{(i)}(x)), \cdot} \tilde{n}) + \tilde{n}_x \label{eqn:cb 4} \\	
\tilde{n}_x &= \max_i (\hat{n}_x - M_{(i,h^{(i)}(x)), \cdot} \hat{n} + V_i(x)). \label{eqn:cb 7}
\end{align}
Initializing with $\tilde{n} = 0$, repeating these iterations until convergence, and returning either $\hat{n}$ or $\tilde{n}$ as the estimate yields the counter braids estimator. 

\subsection{Supplementary figures}
\begin{figure}[H]
	\begin{tabular}{cc}
	\includegraphics[height=1.8in, width=1.7in]{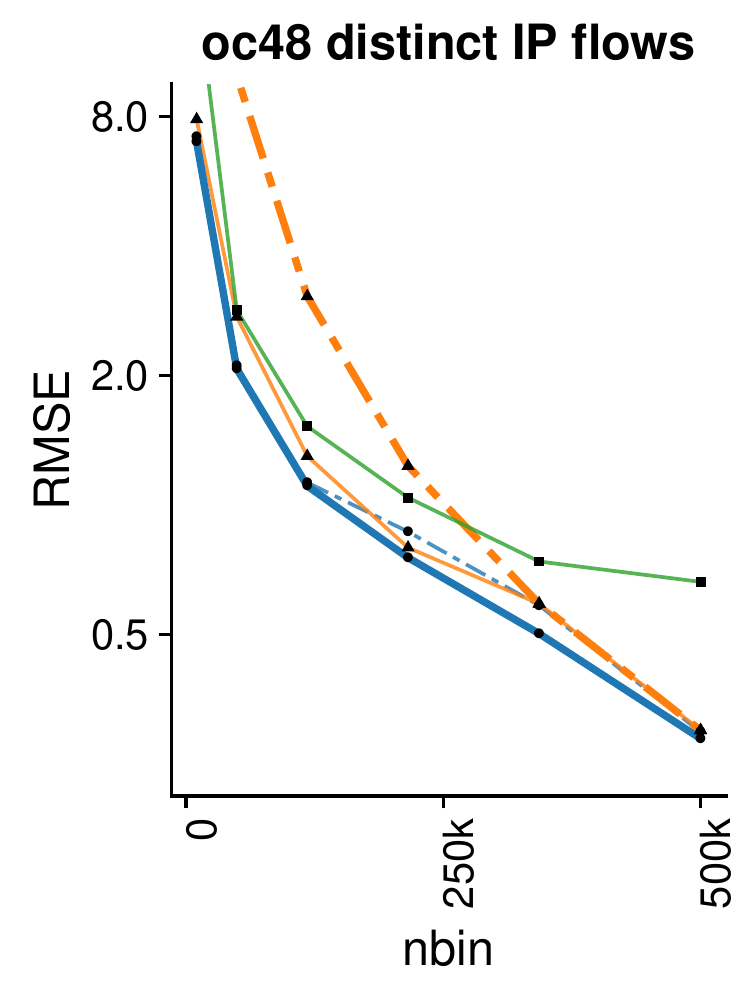} &
		\includegraphics[height=1.8in, width=1.7in]{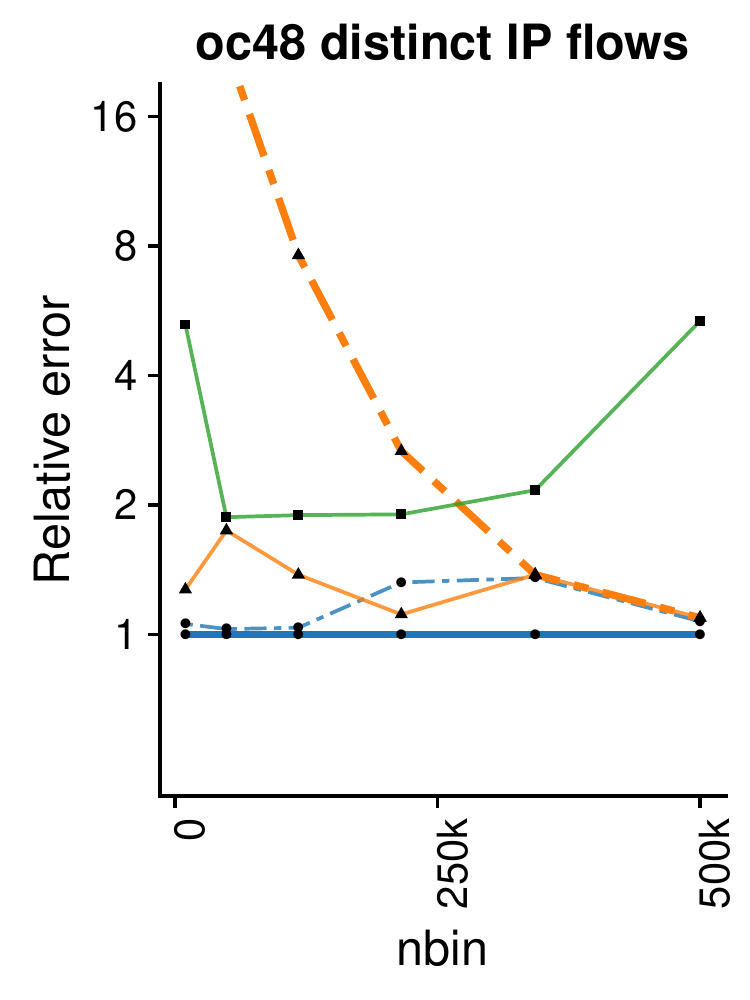} 
	\end{tabular}
\includegraphics[width=3.4in, height=0.2in]{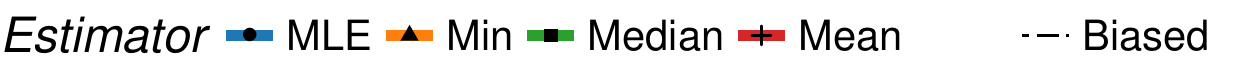}
		\caption{Errors in counts when counting distinct IP flows by source for the oc48 data}
		\label{fig:oc48distinct}
\end{figure}
\end{document}